\documentclass[9pt]{IEEEtran}
\usepackage{cite}
\usepackage{amsmath,amssymb,amsfonts}
\usepackage{algorithmic}
\usepackage{graphicx}
\pdfoutput=1
\usepackage{textcomp}
\def\BibTeX{{\rm B\kern-.05em{\sc i\kern-.025em b}\kern-.08em
    T\kern-.1667em\lower.7ex\hbox{E}\kern-.125emX}}

\usepackage{epstopdf}
\usepackage{pdfpages}
\usepackage{threeparttable} 
\usepackage{multirow} 
\usepackage{booktabs}
\usepackage{enumerate}
\usepackage{algorithm}
\usepackage{algorithmic}
\usepackage{epsfig}
\usepackage{multicol} 
\usepackage{stfloats}

\newtheorem{proposition}{Proposition}
\newtheorem{lemma}{Lemma}

\newtheorem{remark}{Remark}

\newenvironment{proof}{\hspace{0ex}\textsc{Proof}.\hspace{1ex}}{\hfill$\Box$\newline}
\makeatletter

\newcommand{\Rmnum}[1]{\expandafter\@slowromancap\romannumeral #1@}
\usepackage{setspace}

\begin{document}
\title{Stochastic Event-triggered Variational Bayesian Filtering}
\author{Xiaoxu Lv, Peihu Duan, Zhisheng Duan, Guanrong Chen, and Ling Shi
\thanks{
This work is supported by the National Natural Science Foundation of China under grants T2121002 and 6217300.
\emph{(Corresponding author: Zhisheng Duan.)}}

\thanks{
X. Lv, P. Duan and Z. Duan are with the State Key Laboratory for Turbulence and Complex Systems, Department of Mechanics and Engineering Science, College of Engineering, Peking University, Beijing 100871, China. E-mails: lvxx@pku.edu.cn (X. Lv), duanpeihu@pku.edu.cn (P. Duan),
duanzs@pku.edu.cn (Z. Duan).}

\thanks{
 G. Chen is with Department of Electrical Engineering, City University of Hong Kong, Hong Kong SAR, China. E-mail: eegchen@cityu.edu.hk (G. Chen).}

\thanks{
L. Shi is with Department of Electronic and Computer Engineering, the Hong Kong University of Science and Technology, Clear Water
Bay, Kowloon, Hong Kong, China. E-mail: eesling@ust.hk (L. Shi)}}
\maketitle

\begin{abstract}
  This paper proposes an event-triggered variational Bayesian filter for remote state estimation with unknown and time-varying noise  covariances. After presetting multiple nominal process noise covariances and an initial measurement noise covariance, a variational Bayesian method  and a fixed-point iteration method  are utilized to jointly  estimate the posterior state vector and the unknown noise covariances under a stochastic event-triggered mechanism. The proposed  algorithm  ensures low communication loads  and
  excellent estimation performances for a wide range of unknown noise covariances. Finally, the performance of the proposed algorithm is demonstrated by  tracking simulations of a vehicle.

\end{abstract}


\section{Introduction}

The last few decades have seen a great development of state estimation techniques with their wide applications  in navigation, tracking, and sensor networks. Various types of
state estimation algorithms have been proposed such as
linear filters\cite{kalman1960new},  nonlinear filters\cite{julier2000new}\cite{arasaratnam2009cubature}, particle filers\cite{arulampalam2002tutorial}, filters with state constraints\cite{simon2002kalman,lv2021distributed,lv2021distributed2}, filters with intermittent observations\cite{sinopoli2004kalman,kluge2010stochastic}, distributed filters\cite{battistelli2014kullback,duan2020distributed}, etc.

\vspace{4pt}
 In practice,  considering small batteries equipped by  sensors with  limited channel bandwidths, an efficient remote estimation system is desirable. In such a system, the sensor decides whether it sends  measurement to a remote estimator. Usually, a tradeoff between the communication rate and the estimation performance is necessary. Various event-based communication schemes  provide a  solution to such problems, which are typically subject to limited transmission resources as depicted in Fig.\ref{fig communication diagram}.
  An event-based sensor data scheduler and a  state estimation algorithm were proposed in \cite{wu2012event}. However, to drive the minimum mean-squared error (MMSE) estimator,   Gaussian approximation was adopted for prior estimates, and an exact MMSE estimation algorithm was designed by exploiting the generalized closed skew normal distribution in \cite{he2018event}.  Then, two new  types of  event-triggered schedules were designed to eliminate an associated approximation problem \cite{han2015stochastic}. Lately, a robust event-triggered remote state estimation algorithm was derived by minimizing a risk-sensitive criterion in \cite{huang2019robust}. In the presence of packet drops between the sensor and the estimator,  the remote state estimation problem was studied,  and a suboptimal estimator was proposed in \cite{xu2020remote}. For general non-Gaussian systems, an event-based transmission scheme was derived for particle filter based on remote state estimation in \cite{li2019information}.

\begin{figure}[!htb]
\centering
{\includegraphics[width=3.5in]{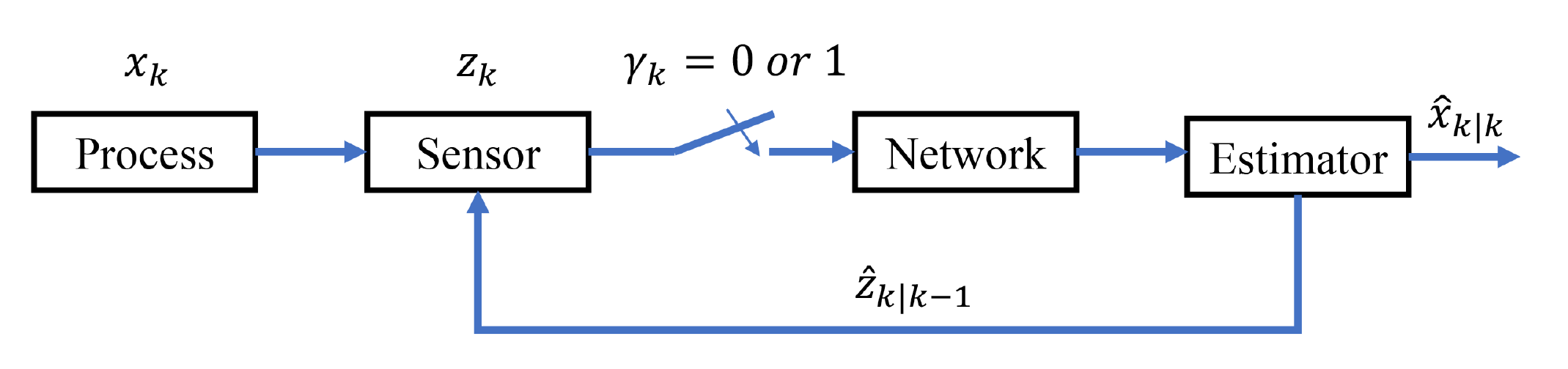}}
\caption{Event-triggered sensor scheduling framework.}
\label{fig communication diagram}
\end{figure}

\vspace{4pt}
  However, in practical applications, the   measurement and the process    noise covariances cannot be precisely calculated, which may even be time-varying.
  To jointly estimate the state and the time-varying noise covariances,
  the variational Bayesian method is effective.
    In \cite{sarkka2009recursive}, a recursive adaptive  Kalman filter was proposed by forming a separable variational approximation based on the inverse-Gamma distribution. Later, by combining  variational Bayesian and Gaussian filtering methods, a variational Bayesian adaptive algorithm was designed in \cite{hartikainen2013variational}. However, the above two filters cannot handle the case with unknown process noise covariances. Recently, by using inverse Wishart priors for the predicted error covariance matrix, a variational Bayesian adaptive Kalman filter was presented in \cite{huang2017novel}. Moreover, a new variational adaptive Kalman filter with the Gaussian-inverse-Wishart mixture distribution was developed in \cite{huang2020variational}.
   By modelling the probability density functions of state transition and measurement as Gaussian-Gamma mixture distributions, an adaptive Kalman filter was derived in \cite{zhu2021adaptive}.
   Although these variational Bayesian algorithms have good filtering performances, they are costly in performing information transmission from a sensor to a remote estimator. It is therefore desirable  to solve the above Bayesian filtering problem with lower transmission costs, which is the objective of the present paper.

\vspace{4pt}
Specifically, this paper addresses the event-triggered variational Bayesian estimation problem with unknown and time-varying noise covariances.
Our main contributions are briefly summarized as follows:
\begin{enumerate}
\item  For remote state estimation with unknown and time-varying noise covariances, this paper proposes an event-triggered variational Bayesian filter to jointly estimate the state and the  noise covariances.
    Multiple inverse Wishart priors are utilized for estimating the predicted error covariance with
     weight combinations  being adaptively inferred by the variational Bayesian method.

\item With an event-triggered mechanism, the algorithm ensures excellent and robust performances for a wide range of unknown noise covariances with low communication overhead.

\end{enumerate}

\vspace{4pt}
 The remainder of this paper is organized as follows. Section \uppercase\expandafter{\romannumeral2} presents preliminaries  and the problem formulation.
 Section \uppercase\expandafter{\romannumeral3} describes the designed variational Bayesian filter. Section \uppercase\expandafter{\romannumeral4} provides   simulations to verify the effectiveness of the proposed algorithm. Section
 \uppercase\expandafter{\romannumeral5} concludes the investigation.

%
%

\section{Preliminaries and Problem Statement}

\subsection{Preliminaries}
Let $P$ be a symmetric  positive definite matrix of random variables, and  define the inverse Wishart distribution as
\begin{equation}
\begin{aligned}
&\text{IW}(P| g,G) =\frac{|G|^{\frac{ g}{2}}|P|^{\frac{-( g+n+1)}{2}}e^{-\frac{1}{2}tr( GP^{-1})}}{2^{\frac{ g n }{2}}\Gamma_n(\frac{g}{2})},
\end{aligned}
\end{equation}
where $n$, $g$, and $G$  denote the dimensions of $P$,
the degree of freedom (dof), a positive definite scale matrix, respectively, and $\Gamma_n$ is the $n$-th order gamma function.



\vspace{4pt}
Let $\lambda$ be the $M$-dimensional binary variable with $\lambda_{j}\in\{0,1\}$ such that $\sum_{j=1}^{M}\lambda_{j} =1$. The categorical distribution is defined as
\begin{equation}\label{eq categorical distribution}
\begin{aligned}
\text{Cat}(\lambda|\mu,M)=\prod_{j=1}^M [\mu_{j}]^{\lambda_{j}},
\end{aligned}
\end{equation}
where $\mu = (\mu_1,...,\mu_M)$ is subject to constraints $0\leq \mu_j\leq 1$ and $\sum_{j=1}^{M} \mu_j =1$.
Moreover,
the Dirichlet distribution is defined as
\begin{equation} \label{eq dirichlet distribution}
\begin{aligned}
 \text{Dir}(\mu,\alpha,M)= C(\alpha)\prod_{j=1}^{M}[\mu_{j}]^{\alpha_{j}-1},
\end{aligned}
\end{equation}
where $\alpha = (\alpha_1,...,\alpha_M)$, $\hat \alpha = \sum_{j=1}^M\alpha_j$, $C(\alpha) = \frac{\Gamma(\hat\alpha)}{\Gamma(\alpha_{1})\cdot \cdot \cdot \Gamma(\alpha_{M})}$,  $E\{\mu_j\}= \frac{\alpha_j}{\hat\alpha}$, and $\Gamma(\cdot)$ is the gamma function. The above distributions will be used for the estimation of covariances later.

\vspace{4pt}
Variational inference is a method based on  optimization to estimate unknown probability densities.
Denote the set of all latent unknown variables by $\Lambda =\{\theta_1,..., \theta_N\}$ and the set of all observed variables by $Z=\{Z_1,..., Z_M\}$.  According to the mean field theory \cite{parisi1988statistical}, the true distribution $p(\Lambda|Z)$ can be approximated by $p(\Lambda|Z) \approx \prod_{k=1}^{N} q(\theta_k) $, where every $\theta_k$ has its own independent distribution  $q(\theta_k)$. Moreover, $q(\theta_k)$ can be designed by minimizing the Kullback-Leibler (KL) divergence between the true distribution and the  approximation as follows \cite{2006Pattern}:
\begin{equation}\label{eq KL}
\begin{aligned}
\text{KL}\left[\prod_{k=1}^{N} q(\theta_k)||p(\Lambda|Z)\right] &= \int \prod_{k=1}^{N} q(\theta_k)\text{log}\left(\frac{\prod_{k=1}^{N} q(\theta_k)}{p(\Lambda|Z)}\right)   d\Lambda\\
 & = -\mathcal{L}\left(\prod_{k=1}^{N} q(\theta_k)\right)+\text{log}~p(Z),
\end{aligned}
\end{equation}
where $\mathcal{L}\left(\prod_{k=1}^{N} q(\theta_k)\right)$ is the evidence lower bound,  formulated as
\begin{equation}\label{eq ELBO}
\begin{aligned}
\mathcal{L}\left(\prod_{k=1}^{N} q(\theta_k)\right) &= \int \text{log}~p(\Lambda,Z)\prod_{k=1}^{N} q(\theta_k) d\theta_k\\
&~~~~-\sum_{k=1}^{N}\int q(\theta_k)\text{log}~q(\theta_k)d\theta_k.
\end{aligned}
\end{equation}

\vspace{4pt}
Based on the variational inference theory,
minimizing the KL divergence (\ref{eq KL}) is equivalent to maximizing the evidence lower bound (\ref{eq ELBO}), which yields  the optimal solution
 \begin{equation}\label{eq max ELBO}
 q^{*}(\theta_k) = \text{arg max}_{q(\theta_k)}\mathcal{L}\left(\prod_{k=1}^{N} q(\theta_k)\right).
 \end{equation}
 This optimal solution  is calculated as \cite{2006Pattern}
 \begin{equation}\label{eq compute KL}
 \text{log}~q^{*}(\theta_k) = E^{*}_{\Lambda^{-\theta_k}}\{\text{log} ~p(\Lambda,Z)\} + c_{\theta_k},
 \end{equation}
where $\Lambda^{-\theta_k}$ denotes the set of all unknown variables $\Lambda$ except $\theta_k$,  and $c_{\theta_k}$ is a constant independent of $\theta_k$. When  these variational parameters are coupled,  the fixed-point iteration method \cite{hildebrand1987introduction} can be adopted to obtain an approximate optimal solution for (\ref{eq KL}) as
 \begin{equation}\label{eq compute KL iteration}
 \text{log}~q^{i+1}(\theta_k) = E^{i}_{\Lambda^{-\theta_k}}\{\text{log}~ p(\Lambda,Z)\} + c_{\theta_k},
 \end{equation}
where $i$ denotes the $n$-th iteration.

\vspace{4pt}
\begin{lemma}\label{lemma block inverse matrix}
For matrices $A$, $B$, $C$, and $D$  with appropriate dimensions, if $A$ and $E=D-CA^{-1}B$ are nonsingular, then
\begin{equation}
\begin{aligned}
\left[\begin{array}{cc} A & B \\ C
& D  \end{array}\right]^{-1} =
\left[\begin{array}{cc} A^{-1} +A^{-1}BE^{-1}CA^{-1} & -A^{-1}BE^{-1} \\  -E^{-1}CA^{-1}
& E^{-1}  \end{array}\right].
\end{aligned}
\end{equation}
\end{lemma}

\subsection{Problem Statement}
Consider a discrete-time system as follows:
\begin{equation}\label{eq dynamical equation}
x_k = F_kx_{k-1}+\omega_k,
\end{equation}
\begin{equation}\label{eq measurement equation}
z_k = H_kx_k+\nu_k,
\end{equation}
where $x_k\in \mathbb{R}^n$ is the system state, $F_k \in \mathbb{R}^{n\times n}$ is the state transition matrix, $z_k\in\mathbb{R}^m$ is the measurement, $H_k\in \mathbb{R}^{m\times n}$ is the measurement matrix, and $\omega_k \in \mathbb{R}^{n}$  and  $\nu_k \in \mathbb{R}^m$ are mutually uncorrelated zero-mean Gaussian noise with covariances $Q_k>0$ and $R_k>0$, respectively.
The initial state  $x_0$ and its estimates $\hat x_{0|0}$ obey Gaussian distributions.

\vspace{4pt}
 In practice, the process and the  measurement noise covariances may be unknown and time-varying, which are usually empirically estimated.
 In this paper, it is assumed that the initial nominal measurement noise covariance and  $M$ nominal process noise covariances are described by
$\bar R_0$ and
\begin{equation}\label{eq Q}
\begin{aligned}
\bar Q_k = \prod_{j=1}^{M} [\bar Q_{j,k}]^{\lambda_{j,k}},
\end{aligned}
\end{equation}
respectively, where $\bar Q_{j,k}$ represents the $j$-th nominal process noise covariance at time $k$, and $\lambda_{k} = (\lambda_{1,k},..., \lambda_{M,k})$ is subject to the categorical distribution like (\ref{eq categorical distribution}). Practically, $\lambda_k$ is not exactly known.

\vspace{4pt}
\begin{remark}
Under the event-triggered mechanism, the unknown process noise covariances cannot be jointly estimated in a recursive form like the measurement noise covariances by the variational Bayesian method. To solve this problem, the unknown process noise covariances are expected to be estimated  by adaptive weight combinations  of multiple nominal process noise covariances. From (\ref{eq Q}),
$E\{\bar Q_k\}=\sum_{j=1}^{M}E\{\lambda_{j,k}\}\bar Q_{j,k}$, where $E\{\lambda_{j,k}\}$ is utilized as the adaptive weight parameter.
 Hence, the unknown process noise covariance estimation problem
is reformulated  as  the adaptive weight combination problem of  multiple nominal process noise covariances.
\end{remark}

\vspace{4pt}
In the present framework, when the sensor obtains the measurement $z_k$, to reduce the transmission cost, it first makes the decision whether to send the measurement $z_k$ to a remote estimator.  To do so, a binary decision variable $\gamma_k=1$ or $0$ is introduced. Specifically, if the sensor decides to send the measurement, $\gamma_k=1$; otherwise $\gamma_k =0$.
 Here, the  information set at time $k$ for the estimator is defined  as
$\mathcal{I}_k \triangleq \{ \gamma_0z_0,...,\gamma_kz_k\}\cup\{\gamma_0,...,\gamma_k\}$
 with $\mathcal{I}_{-1}=\emptyset$.
 Here, a closed-loop stochastic event-triggered schedule is adopted to design the decision variable $\gamma_k$ as [15]
\begin{equation}\label{eq event-triggered gamma}
\gamma_k = \left\{
\begin{array}{cc}
0, & \zeta \leq \varphi(e_k)\\
1, &  \zeta > \varphi (e_k),\\
\end{array}\right.
\end{equation}
where $\zeta$ is a uniformly distributed random variable over $[0,1]$ at every step;
$e_k = z_k - \hat z_{k|k-1}$ is the innovation with the feedback predicted measurement $\hat z_{k|k-1}=H_k\hat x_{k|k-1}$; $Y_k$ is a positive definite matrix; and $\varphi(e_k)$ is designed as
\begin{equation}\label{eq event sigma}
\varphi(e_k) = \text{exp}\left(-\frac{1}{2}e^T_k Y_k e_k\right).
\end{equation}

\vspace{4pt}
The objective of this paper is to design an event-triggered variational Bayesian filter for the sensor system with  unknown  process and measurement noise covariances. The following tasks will be accomplished:
\begin{enumerate}
\item design  a variational Bayesian filter in the presence  of unknown noise covariances with multiple nominal covariances.

\item  design the above filter with a low transmission cost using  the event-triggered schedule (\ref{eq event-triggered gamma}).
\end{enumerate}

\section{Variational Bayesian  filter}
 In this section, a novel filtering algorithm is proposed  by using the variational Bayesian method, estimating the unknown states $x_k$ under the event-triggered mechanism.

\vspace{4pt}
\subsection{Unknown Noise Distributions}
First, it follows from (\ref{eq Q}) that  the prior state $\hat x_{k|k-1}$ and the nominal predicted error covariance $P_{k|k-1}$ are given by
\begin{equation}\label{eq predicted state}
\begin{aligned}
\hat x_{k|k-1} &= F_{k-1}\hat x_{k-1|k-1},
\end{aligned}
\end{equation}
and
\begin{equation}\label{eq predicted covariances}
\begin{aligned}
P_{k|k-1} &= \prod_{j=1}^{M} [ F_{k-1}P_{k-1|k-1}F^T_{k-1}+\bar Q_{j,k}  ]^{\lambda_{j,k}}\\
& =\prod_{j=1}^{M}[P_{j,k|k-1}]^{\lambda_{j,k}},
\end{aligned}
\end{equation}
respectively, where $P_{j,k|k-1}= F_{k-1}P_{k-1|k-1}F^T_{k-1}+\bar Q_{j,k}$, and $\lambda_{j,k}$ is given in (\ref{eq Q}). However, since the noise covariances are unknown, it is not feasible to directly perform (\ref{eq predicted state}) and (\ref{eq predicted covariances}) for state estimation. To address this issue, the noise covariances are estimated as follows.

\vspace{4pt}
  Following \cite{2006Pattern}, the inverse Wishart distribution is utilized as the conjugate prior distribution for the Gaussian distribution with an unknown covariance matrix. In so doing, the  distributions of $p(P_{k|k-1}|\mathcal{I}_{k-1})$ and $p(R_k|I_{k})$ are chosen as
\begin{equation}\label{eq PIWlambda}
\begin{aligned}
&p(P_{k|k-1}|\lambda_k, \mathcal{I}_{k-1})=\prod _{j=1}^{M}\left[ \text{IW}(P_{k|k-1}|\hat g_{j,k|k-1},\hat G_{j,k|k-1})\right]^{\lambda_{j,k}}
\end{aligned}
\end{equation}
and
\begin{equation}
\begin{aligned}
&p(R_{k-1}| \mathcal{I}_{k-1})=\text{IW}(R_{k-1}|\hat s_{k-1|k-1},\hat S_{k-1|k-1}),
\end{aligned}
\end{equation}
respectively, where $\hat g_{j,k|k-1}$ is a preselected tuning parameter, $\hat G_{j,k|k-1}=\hat g_{j,k|k-1}P_{j,k|k-1}$, $\hat s_{0|0}$ is another preselected tuning parameter, and $\hat S_{0|0}=\hat s_{0|0}R_{0|0}$. It should be noted that $\lambda_k$ in (\ref{eq PIWlambda}) are unknown.

\vspace{4pt}
To obtain the value for $\lambda_k$ in (\ref{eq PIWlambda}), following \cite{huang2020variational}, it is modelled
 by the categorical distribution, like (\ref{eq categorical distribution}), as
\begin{equation*}
\begin{aligned}
p(\lambda_{k}|\mu_k)& = \text{Cat}(\lambda_k|\mu_k,M)=\prod_{j=1}^M [\mu_{j,k}]^{\lambda_{j,k}}.
\end{aligned}
\end{equation*}
Then, since the conjugate prior distribution of the categorical distribution is the Dirichlet distribution \cite{2006Pattern},
 $\mu_k$ is  modelled by the Dirichlet distribution, like (\ref{eq dirichlet distribution}), as
\begin{equation*}
\begin{aligned}
p(\mu_k)& = \text{Dir}(\mu_k,\alpha_{k|k-1},M)\\
&= C(\alpha_{k|k-1})\prod_{j=1}^{M}[\mu_{j,k}]^{\alpha_{j,k|k-1}-1}.
\end{aligned}
\end{equation*}
 To this end, the distributions of the unknown noise covariances have been constructed. Based on these distributions,   the posterior estimates will be inferred  by the variational Bayesian method, as further discussed in the following.

\subsection{Variational Approximation: $\gamma_k = 0$}
This subsection  studies the case where $z_k$ is not transmitted by the sensor at step $k$ under the event-triggered law (\ref{eq event-triggered gamma}), which results in the unavailability of $z_k$.
Under this circumstance,  $x_k$ and $z_k$ are regarded as jointly Gaussian \cite{han2015stochastic}.
Hence, they are strongly coupled and have to be jointly estimated. It is worth mentioning that the coupling poses a great challenge to obtain the solutions  of   $x_k$, $P_{k|k-1}$, and $R_k$ by using the traditional variational Bayesian technique, since $P_{k|k-1}$, and $R_k$ are coupled in the  logarithm of the joint  probability density function. To address this difficulty,  a novel decoupling approach is proposed as follows.

\vspace{4pt}
For convenience, denote the set of the unknown variables as
\begin{equation}\label{eq lambda0}
\begin{aligned}
\Lambda_0 \triangleq \{(x_k,z_k) ,P_{k|k-1},R_k,\lambda_k,\mu_k\}.
\end{aligned}
\end{equation}

\vspace{4pt}
In the following, the approximate posterior  probability density function (PDF) for every element in $\Lambda_0$ is calculated.

\vspace{4pt}
\subsubsection{ The logarithm of joint PDF $p(\Lambda_0,\gamma_k)$}
First, the joint PDF $p(\Lambda_0,\gamma_k)$ is factorized as
\begin{equation}\label{eq p lambda0 gamma 1}
\begin{aligned}
p(\Lambda_0,\gamma_k|\mathcal{I}_{k-1})=&p(\gamma_k|z_k,Y_k,\mathcal{I}_{k-1})
 p(x_k,z_k|P_{k|k-1},R_k,\mathcal{I}_{k-1})\\
&\times p(P_{k|k-1}|\lambda_k)p(\lambda_k|\mu_k)p(\mu_k)
p(R_k).
\end{aligned}
\end{equation}

\vspace{4pt}
Based on  the event-triggered scheme (\ref{eq event-triggered gamma}), one has
\begin{equation*}
\begin{aligned}
&p(\gamma_k=0|z_k,Y_k,\mathcal{I}_{k-1})\\
 &= Pr(\text{exp}(-\frac{1}{2}e^T_kY_ke_k)\geq \zeta_k|z_k,Y_k,\mathcal{I}_{k-1})\\
 & =\text{exp}(-\frac{1}{2}(z_k-H_k\hat x_{k|k-1})^TY_k(z_k-H_k\hat x_{k|k-1})).
\end{aligned}
\end{equation*}

\vspace{4pt}
Define  $\phi_{k} \triangleq  [x^T_{k}, z^T_{k}]^T$,  $[\hat x^T_{k|k-1}, \hat z^T_{k|k-1}]^T \triangleq E\{ \hat \phi_{k|k-1}|\mathcal{I}_{k-1}\}$, and $\Phi_{k|k-1}\triangleq E\{ (\hat\phi_k - \phi_k)(\hat\phi_k -\phi_k)^T|\mathcal{I}_{k-1} \}$. Then, $\Phi_{k|k-1}$ is obtained as
\begin{equation*}
\begin{aligned}
&\Phi_{k|k-1} = \left[\begin{array}{cc} P_{k|k-1}&P_{k|k-1}H^T_k\\H_kP_{k|k-1}&H_kP_{k|k-1}H^T_k+R_k  \end{array}\right].
\end{aligned}
\end{equation*}
Hence, the probability density function of $\phi_k$ in (\ref{eq p lambda0 gamma 1}) is given by
\begin{equation}\label{eq p(x,z)}
\begin{aligned}
&p(x_k,z_k|P_{k|k-1},R_k, \mathcal{I}_{k-1})=\mathcal{N}(\phi_{k}|\hat \phi_{k|k-1},
\Phi_{k|k-1})\\
&=\frac{ \text{exp}\left(-\frac{1}{2} \left[\begin{array}{c} x_{k}-\hat x_{k|k-1}\\z_{k}-\hat z_{k|k-1}  \end{array}\right]^T\Phi^{-1}_{k|k-1} \left[\begin{array}{c} x_{k}-\hat x_{k|k-1}\\z_{k}-\hat z_{k|k-1}  \end{array}\right]\right)}{(2\pi)^{(n+m)/2}|\Phi_{k|k-1}|^{1/2}}.
\end{aligned}
\end{equation}

\vspace{4pt}
Using the above distributions, $p(\Lambda_0,\gamma_k)$ in (\ref{eq p lambda0 gamma 1})  is reformulated as
\begin{equation} \label{eq p(lambda0)gamma}
\begin{aligned}
 p(\Lambda_0,\gamma_k|\mathcal{I}_{k-1})=&p(\gamma_k|z_k,Y_k,\mathcal{I}_{k-1}) \mathcal{N}(\phi_{k|k-1}|\hat \phi_{k|k-1},
\Phi_{k|k-1})\\
&\times \prod _{j=1}^{M}\left[ \text{IW}(P_{k|k-1}|\hat g_{j,k|k-1},\hat G_{j,k|k-1})\right]^{\lambda_{j,k}}\\
&\times \text{Cat}(\lambda_k|\mu_k,M) \text{Dir}(\mu_k,\alpha_{k|k-1},M)\\
&\times \text{IW}(R_{k}|\hat s_{k|k-1},\hat S_{k|k-1}).\\
\end{aligned}
\end{equation}

\vspace{4pt}
Then, using (\ref{eq compute KL iteration}), $\text{log} ~ p(\Lambda_0,\gamma_k|\mathcal{I}_{k-1})$ in (\ref{eq p(lambda0)gamma}) is  calculated as
\begin{equation}\label{eq log pLambda gamma}
\begin{aligned}
&\text{log} ~ p(\Lambda_0,\gamma_k|\mathcal{I}_{k-1})\\
=& -\frac{1}{2}(z_k-H_k\hat x_{k|k-1})^TY_k(z_k-H_k\hat x_{k|k-1})-0.5\text{log}~|\Phi_{k|k-1}|\\
&+\sum_{j=1}^M  \lambda_{j,k} \text{log}\mu_{j,k}-\frac{1}{2}(\phi_{k}-\hat \phi_{k|k-1})^T\Phi^{-1}_{k|k-1}(\phi_{k}-\hat \phi_{k|k-1})\\
&+\sum_{j=1}^M \lambda_{j,k} \bigg[ 0.5 \hat g_{j,k|k-1} \text{log}(|\hat G_{j,k|k-1}|)-0.5\text{tr}(\hat G_{j,k|k-1}P^{-1}_{k|k-1})\\
&-0.5(\hat g_{j,k|k-1}+n+1)\text{log}(|P_{k|k-1}|)-0.5(n\hat g_{j,k|k-1})\text{log}2\\
&-\text{log}\Gamma_n(\hat g_{j,k|k-1}/2) \bigg]+\sum_{j=1}^M (\alpha_{j,k|k-1}) \text{log}\mu_{j,k} \\
&-0.5(\hat s_{k|k-1}+m+1)\text{log}~|R_{k}|-0.5\text{tr}(\hat S_{k|k-1}R^{-1}_{k}) + c_{\Lambda_0},
\end{aligned}
\end{equation}
 where $c_{\Lambda_0}$ is a constant independent of $\Lambda_0$.

\vspace{4pt}
\subsubsection{Decoupling of $\Phi_{k|k-1}$}
Each element of $\Lambda_0$ in (\ref{eq p(lambda0)gamma}) will be calculated alternatively, and the corresponding posterior distribution should have the same functional form as the prior distribution according to the variational inference theory \cite{2006Pattern}. However, $P_{k|k-1}$ and $R_k$ are coupled in $\Phi_{k|k-1}$,
which poses a challenge in deriving the posterior $x_k$, $P_{k|k-1}$, and $R_k$ by using (\ref{eq compute KL iteration}).  In (\ref{eq log pLambda gamma}),
$-0.5\text{log}~|\Phi_{k|k-1}|$ and $-\frac{1}{2}(\phi_{k}-\hat \phi_{k|k-1})^T\Phi^{-1}_{k|k-1}(\phi_{k}-\hat \phi_{k|k-1})$ will be decoupled for $P_{k|k-1}$ and $R_k$.
 In addition, the expectation of $\Phi^{-1}_{k|k-1}$ is calculated, as follows.

\vspace{4pt}
First,  $\Phi_{k|k-1}$ is factorized as
\begin{equation}\label{eq phi}
\begin{aligned}
\Phi_{k|k-1}=\left[ \begin{array}{cc}I & 0\\ H& I\end{array}\right]
\left[ \begin{array}{cc}P_{k|k-1} & 0\\ 0& R_{k}\end{array}\right]
\left[ \begin{array}{cc}I & H^T\\ 0& I\end{array}\right].
\end{aligned}
\end{equation}
Then, the inverse of $\Phi_{k|k-1}$ is computed by
\begin{equation}\label{eq phi-1}
\begin{aligned}
\Phi^{-1}_{k|k-1}=\left[ \begin{array}{cc}I & -H^T\\ 0& I\end{array}\right]
\left[ \begin{array}{cc}P^{-1}_{k|k-1} & 0\\ 0& R^{-1}_{k}\end{array}\right]
\left[ \begin{array}{cc}I & 0\\ -H& I\end{array}\right].
\end{aligned}
\end{equation}
  It follows from (\ref{eq phi}) that the determinant of $\Phi_{k|k-1}$ is obtained as
\begin{equation}\label{eq |phi|}
\begin{aligned}
|\Phi_{k|k-1}| =|P_{k|k-1}||R_{k}|.
\end{aligned}
\end{equation}
Hence,
\begin{equation}\label{eq ln|phi|}
\begin{aligned}
\text{log}|\Phi_{k|k-1}| =\text{log}|P_{k|k-1}|+\text{log} |R_{k}|.
\end{aligned}
\end{equation}

\vspace{4pt}
Now, based on (\ref{eq phi-1}) and (\ref{eq p(x,z)}), define a new variable as
\begin{equation*}
\begin{aligned}
\rho_{k|k-1} =\left[ \begin{array}{cc}I & 0\\ -H& I\end{array}\right]\left[\begin{array}{c} x_{k}-\hat x_{k|k-1}\\z_{k}-\hat z_{k|k-1}  \end{array} \right].
\end{aligned}
\end{equation*}
Then, $\rho_{k|k-1}\rho^T_{k|k-1}$ is computed by
\begin{equation}\label{eq rhorho^T}
\begin{aligned}
&~~~~\rho_{k|k-1}\rho^T_{k|k-1}\\
&=\left[ \begin{array}{cc}I & 0\\ -H& I\end{array}\right](\phi_{k}-\hat\phi_{k|k-1})(\phi_{k}-\hat\phi_{k|k-1})^T\left[ \begin{array}{cc}I & 0\\ -H& I\end{array}\right]^T\\
& = \left[ \begin{array}{cc}
(\rho_{k|k-1}\rho^T_{k|k-1} )_{xx} & (\rho_{k|k-1}\rho^T_{k|k-1} )_{xz}\\
(\rho_{k|k-1}\rho^T_{k|k-1} )_{zx} &
(\rho_{k|k-1}\rho^T_{k|k-1} )_{zz}\\
\end{array} \right],
\end{aligned}
\end{equation}
where $(\rho_{k|k-1}\rho^T_{k|k-1} )_{xx}$ and $(\rho_{k|k-1}\rho^T_{k|k-1} )_{zz}$ are calculated as
\begin{equation}
\begin{aligned}
(\rho_{k|k-1}\rho^T_{k|k-1} )_{xx}=(x_{k}-\hat x_{k|k-1}) (x_{k}-\hat x_{k|k-1})^T
\end{aligned}
\end{equation}
and
\begin{equation}\label{eq rhorhozz}
\begin{aligned}
(\rho_{k|k-1}\rho^T_{k|k-1} )_{zz}=&H_k(x_k-\hat x_{k|k-1})(x_k-\hat x_{k|k-1})^TH^T_k\\
&-(z_k-\hat z_{k|k-1})(x_k-\hat x_{k|k-1})^TH^T\\
& -H_k(x_k-\hat x_{k|k-1})(z_k-\hat z_{k|k-1})\\
& +(z_k-\hat z_{k|k-1})(z_k-\hat z_{k|k-1})^T,
\end{aligned}
\end{equation}
respectively.

\vspace{4pt}
Combined (\ref{eq phi-1})-(\ref{eq rhorhozz}), one has
\begin{equation}\label{eq phiPhiphi}
\begin{aligned}
&\frac{1}{2}(\phi_{k}-\hat \phi_{k|k-1})^T\Phi^{-1}_{k|k-1}(\phi_{k}-\hat \phi_{k|k-1})\\
& =\frac{1}{2}\left[\begin{array}{c} x_{k}-\hat x_{k|k-1}\\z_{k}-\hat z_{k|k-1}  \end{array}\right]^T\Phi^{-1}_{k|k-1}\left[\begin{array}{c} x_{k}-\hat x_{k|k-1}\\z_{k}-\hat z_{k|k-1}  \end{array}\right]\\
&=0.5\rho^T_{k|k-1} \left[ \begin{array}{cc}P^{-1}_{k|k-1} & 0\\ 0& R^{-1}_{k}\end{array}\right]\rho_{k|k-1}\\
& = 0.5\text{tr}\left( \rho_{k|k-1}\rho^T_{k|k-1} \left[ \begin{array}{cc}P^{-1}_{k|k-1} & 0\\ 0& R^{-1}_{k}\end{array}\right] \right)\\
& =0.5\text{tr}((\rho_{k|k-1}\rho^T_{k|k-1} )_{xx}P^{-1}_{k|k-1})+0.5\text{tr}((\rho_{k|k-1}\rho^T_{k|k-1} )_{zz}R^{-1}_{k}).
\end{aligned}
\end{equation}
 Hence, according to (\ref{eq ln|phi|}) and (\ref{eq phiPhiphi}), $P_{k|k-1}$ and $R_k$ are decoupled. Then, the approximate optimal solutions $q(P_{k|k-1})$ and $q(R_k)$ can be calculated by (\ref{eq compute KL iteration}), respectively. Furthermore,  the  posterior distributions are derived in  the same form as the prior distribution to  provide a  recursive filter.

\subsubsection{The update of $\phi_k$}
Now, every component in $\Lambda_0$ is computed.
Let $\theta = \phi_k$, so that
\begin{equation}\label{eq q(x,z)}
\begin{aligned}
&~~~~\text{log} ~ q^{i+1}(\phi_k) \\
&= -\frac{1}{2}(z_k-H_k\hat x_{k|k-1})^T Y_k(z_k-H_k\hat x_{k|k-1})\\
& ~~~~-\frac{1}{2}(\phi_{k}-\hat \phi_{k|k-1})^TE^i\{\Phi^{-1}_{k|k-1} \}(\phi_{k}-\hat \phi_{k|k-1})+c_{x_k, z_k}\\
& = -\frac{1}{2}(\phi_{k}-\hat \phi_{k|k-1})^T(\Theta^{i+1}_k )^{-1}(\phi_{k}-\hat \phi_{k|k-1}) +c_{x_k, z_k},
\end{aligned}
\end{equation}
where  $c_{x_k, z_k}$ is a constant independent of $x_k, z_k$, and
$\Theta^{i+1}_k$  is given by \begin{equation}\label{eq Theta}
\begin{aligned}
\Theta^{i+1}_k &= \left[\begin{array}{cc} P^{i+1}_{k|k}& P^{i+1}_{xz,k|k} \\  P^{i+1}_{zx,k|k}
& P^{i+1}_{zz,k|k} \end{array}\right]\\
&= \left(E^i\{\Phi^{-1}_{k|k-1} \} +
\left[\begin{array}{cc}0 & 0 \\  0
& Y_k  \end{array}\right]\right)^{-1}.
\end{aligned}
\end{equation}

\vspace{4pt}
It follows from (\ref{eq phi-1}) that $E^{i}\{ \Phi^{-1}_{k|k-1} \}$ in (\ref{eq Theta}) can be computed by
\begin{equation}\label{eq phi-1 ex}
\begin{aligned}
E^{i}\{ \Phi^{-1}_{k|k-1} \}&=\left[ \begin{array}{cc}I & -H^T\\ 0& I\end{array}\right]
\left[ \begin{array}{cc}E^{i}\{P^{-1}_{k|k-1} \} & 0\\ 0& E^{i}\{R^{-1}_{k} \}\end{array}\right]\\
&~~~~\times\left[ \begin{array}{cc}I & 0\\ -H& I\end{array}\right].\\
\end{aligned}
\end{equation}
 Note that $P_{k|k-1}$ and $R_k$ follow inverse Wishart distributions,  and $P^{-1}_{k|k-1}$ and $R^{-1}_k$ obey  Wishart distributions.
  Define $(\tilde P^{i}_{k|k-1})^{-1} = E^{i}\{P^{-1}_{k|k-1}\}$ and $(\tilde R^{i}_{k})^{-1} = E^{i}\{R^{-1}_{k}\}$. Further, $E^i\{\Phi^{-1}_{k|k-1}\}$ in (\ref{eq phi-1 ex}) can be derived as
\begin{equation*}
\begin{aligned}
&~~~~E^{i}\{ \Phi^{-1}_{k|k-1} \}\\
&=\left[ \begin{array}{cc}I & -H^T\\ 0& I\end{array}\right]
\left[ \begin{array}{cc}(\tilde P^{i}_{k|k-1})^{-1} & 0\\ 0& (\tilde R^{i}_{k})^{-1}\end{array}\right]\left[ \begin{array}{cc}I & 0\\ -H& I\end{array}\right]\\
&=\left[ \begin{array}{cc}(\tilde P^{i}_{k|k-1})^{-1}+H^T_{k}(\tilde R^{i}_{k})^{-1} H_{k} & -H^T_{k}(\tilde R^{i}_{k})^{-1}\\ -(\tilde R^{i}_{k})^{-1}H_{k}& (\tilde R^{i}_{k})^{-1}\end{array}\right]\\
& = (\tilde \Phi^{i}_{k|k-1})^{-1}.
\end{aligned}
\end{equation*}

To this end, $\Theta^{i+1}_k$ can be computed by using Lemma \ref{lemma block inverse matrix}.
Denote $A_{\Theta}= (\tilde P^{i}_{k|k-1})^{-1}+H^T_{k}(\tilde R^{i}_{k})^{-1} H_{k}$, $B_{\Theta} =  -H^T_{k}(\tilde R^{i}_{k})^{-1}$, $C_{\Theta} = -(\tilde R^{i}_{k})^{-1}H_{k}$, and $D_{\Theta} =(\tilde R^{i}_{k})^{-1}+ Y_k$.
 Then, $P^{i+1}_{k|k}$ and $P^{i+1}_{zz,k|k}$ in (\ref{eq q(x,z)}) can be calculated as
\begin{equation}\label{eq Pkk}
\begin{aligned}
P^{i+1}_{k|k}
&= (A_{\Theta} + B_{\Theta}D^{-1}_{\Theta}C_{\Theta})^{-1}\\
& = ((\tilde P^{i}_{k|k-1})^{-1}+H^T_{k}((\tilde R^{i}_{k})^{-1}-(\tilde R^{i}_{k})^{-1}((\tilde R^{i}_{k})^{-1}\\
&~~~~+Y_k)^{-1}(\tilde R^{i}_{k})^{-1}    )H_{k})^{-1} \\
& = ((\tilde P^{i}_{k|k-1})^{-1}+H^T_{k}(\tilde R^{i}_{k}
+ Y^{-1}_k)^{-1}H_k)^{-1}\\
&=\tilde P^{i}_{k|k-1}-\tilde P^{i}_{k|k-1}H^T_k(H_k\tilde P^{i}_{k|k-1}H^T_k\\
&~~~~+\tilde R^{i}_k+Y^{-1}_k)^{-1}H_k\tilde P^{i}_{k|k-1},
\end{aligned}
\end{equation}
and
\begin{equation}\label{eq Pzzkk}
\begin{aligned}
 P^{i+1}_{zz,k|k}&=  ((\tilde R^{i}_k)^{-1}-(\tilde R^{i}_k)^{-1}H_k((\tilde P^{i}_{k|k-1})^{-1}\\
 &~~~~+ H^T_k(\tilde R^{i}_k)^{-1}H_k)^{-1}H^T_k(\tilde R^{i}_k)^{-1}+Y_k )^{-1}\\
&=((H_k\tilde P^{i}_{k|k-1}H^T_k+\tilde R^{i}_k)^{-1}+Y_{k})^{-1},
\end{aligned}
\end{equation}
respectively.

\vspace{4pt}
 Further, additional mathematical  operations are performed to derive a concise form of $P^{i+1}_{xz,k|k}$ as follows:
\begin{equation}\label{eq Pxzkk}
\begin{aligned}
P^{i+1}_{xz,k|k}&=((\tilde P^{i}_{k|k-1})^{-1} + H^T_k(\tilde R^{i}_{k})^{-1}H_k)^{-1}H^T_k(\tilde R^{i}_{k})^{-1}\\
&~~~~\times((H_k\tilde P^{i}_{k|k-1}H^T_k+\tilde R^{i}_{k})^{-1}+Y_k)^{-1}\\
& = (\tilde P^{i}_{k|k-1}-\tilde P^{i}_{k|k-1}H^T_k(\tilde R^{i}_{k}+H_k\tilde P^{i}_{k|k-1}H^T_k)^{-1}H_k\\
&~~~~\times\tilde P^{i}_{k|k-1}) H^T_k(\tilde R^{i}_{k})^{-1}(\tilde R^{i}_{k}+H_k\tilde P^{i}_{k|k-1}H^T_k)\\
&~~~~\times(I+Y_k(\tilde R^{i}_{k}+H_k\tilde P^{i}_{k|k-1}H^T_k))^{-1}\\
& =\tilde P^{i}_{k|k-1}H^T_k(I-(\tilde R^{i}_{k}+H_k\tilde P^{i}_{k|k-1}H^T_k)^{-1}H_k\\
& ~~~~\times \tilde P^{i}_{k|k-1}H^T_k)(\tilde R^{i}_{k})^{-1}(\tilde R^{i}_{k}+H_k\tilde P^{i}_{k|k-1}H^T_k)\\
& ~~~~\times(I+Y_k(\tilde R^{i}_{k}+H_k\tilde P^{i}_{k|k-1}H^T_k))^{-1}\\
& =\tilde P^{i}_{k|k-1}H^T_k(I+Y_k(\tilde R^{i}_{k}+H_k\tilde P^{i}_{k|k-1}H^T_k))^{-1}.
\end{aligned}
\end{equation}
Now, all elements in $\Theta^{i+1}_{k}$ have been obtained. Here, it can be observed that   $\Theta^{i+1}_{k}$ has a similar form  with that in \cite{han2015stochastic}. Hence, from (\ref{eq q(x,z)}), $\phi_k$ follows the Gaussian distribution with mean vector $\hat x^{i+1}_{k|k} = \hat x_{k|k-1}, \hat z^{i+1}_{k|k} = H_k\hat x_{k|k-1}$ and covariance $\Theta^{i+1}_k$.

\vspace{4pt}
\subsubsection{The update of $P_{k|k-1}$}
Let $\theta = P_{k|k-1}$, so that
\begin{equation}\label{eq q(P)}
\begin{aligned}
&~~~~\text{log} ~ q^{i+1}(P_{k|k-1}) \\
&=  -0.5\text{log}(|P_{k|k-1}|)-0.5E^{i+1}\{ \text{tr}((\rho_{k|k-1}\rho^T_{k|k-1} )_{xx}P^{-1}_{k|k-1})\}\\
&~~~~+\sum_{j=1}^M E^{i}\{\lambda_{j,k}\} \bigg[ -0.5\text{tr}(\hat G_{j,k|k-1}P^{-1}_{k|k-1})\\
&~~~~-0.5(\hat g_{j,k|k-1}+n+1)\text{log}(|P_{k|k-1}|)\bigg]+c_P\\
& = -0.5\bigg[\sum_{j=1}^ME^{i}\{\lambda_{j,k}\}\hat g_{j,k|k-1}+n+2\bigg] \text{log}(|P_{k|k-1}|)\\
&~~~~-0.5\text{tr}\bigg[\bigg(A^{i+1}_k+ \sum_{j=1}^ME^i\{\lambda_{j,k}\}\hat G_{k|k-1}\bigg)P^{-1}_{k|k-1}\bigg]+c_P,
\end{aligned}
\end{equation}
where  $c_{P}$ is a constant independent of $P_{k|k-1}$, and
\begin{equation*}
\begin{aligned}
A^{i+1}_k&=E^{i+1}\{ (\rho_{k|k-1}\rho^T_{k|k-1} )_{xx}\}\\
&=P^{i+1}_{k|k}+(\hat x^{i+1}_{k|k}-\hat x_{k|k-1})(\hat x^{i+1}_{k|k}-\hat x_{k|k-1})^T.
\end{aligned}
\end{equation*}
Based on (\ref{eq q(P)}), it can be seen that $P_{k|k-1}$ is the inverse Wishart distribution $\text{IW}(P_{k|k-1}|\hat g^{i+1}_{k|k},\hat G^{i+1}_{k|k})$, where
\begin{equation*}
\begin{aligned}
&\hat g^{i+1}_{k|k}= \sum_{j=1}^ME^i\{\lambda_{j,k}\}\hat g_{j,k|k-1}+1
\end{aligned}
\end{equation*}
and
\begin{equation*}
\begin{aligned}
\hat G^{i+1}_{k|k}= \sum_{j=1}^ME^i\{\lambda_{j,k}\}\hat G_{j,k|k-1}+A^{i+1}_k,
\end{aligned}
\end{equation*}
respectively. Now, according to the property of the inverse Wishart distribution \cite{nydick2012wishart}, $P^{-1}_{k|k-1}$ is the Wishart distribution with $E^{i+1}\{P^{-1}_{k|k-1}\} = \hat g^{i+1}_{k|k} (\hat G^{i+1}_{k|k})^{-1}$, and $E^i\{ \lambda_{j,k} \} = \hat \chi^{i}_k$ will be given later.
For further iteration, define
\begin{equation*}
\begin{aligned}
\tilde P^{i+1}_{k|k-1}&= (E^{i+1}\{P^{-1}_{k|k-1} \})^{-1}  =\hat G^{i+1}_{k|k}/\hat g^{i+1}_{k|k}.
\end{aligned}
\end{equation*}

\vspace{4pt}
\subsubsection{The update of $R_k$} Let $\theta = R_{k}$, so that
\begin{equation*}
\begin{aligned}
  \text{log}~ q^{i+1}(R_{k})& = - 0.5E^{i+1}\{ \text{tr}((\rho_{k|k-1}\rho^T_{k|k-1} )_{zz}R^{-1}_{k})\}\\
& ~~~~-0.5(\hat s_{k|k-1}+m+2)\text{log}~|R_{k}|\\
&~~~~-0.5\text{tr}(\hat S_{k|k-1}R^{-1}_{k}) + c_{R}\\
& = -0.5(\hat s_{k|k-1}+m+2)\text{log}~|R_{k}|\\
&~~~~-0.5\text{tr}((\hat S_{k|k-1}+B^{i+1}_{k})R^{-1}_{k})+ c_{R},
\end{aligned}
\end{equation*}
where  $c_{R}$ is a constant independent of $R_k$, and
\begin{equation*}
\begin{aligned}
B^{i+1}_k&=E^{i+1}\{(\rho_{k|k-1}\rho^T_{k|k-1} )_{zz}\}\\
&=E^{i+1}\{H_k(x_k-\hat x_{k|k-1})(x_k-\hat x_{k|k-1})^TH^T_k\\
&~~~~-(z_k-\hat z_{k|k-1})(x_k-\hat x_{k|k-1})^T\\
&~~~\times H^T-H_k(x_k-\hat x_{k|k-1})(z_k-\hat z_{k|k-1}) \\
&~~~+(z_k-\hat z_{k|k-1})(z_k-\hat z_{k|k-1})^T\}\\
&= H_kP^{i+1}_{k|k}H^T_k - (H_kP^{i+1}_{xz,k|k})^T -H_kP^{i+1}_{xz,k|k}+P^{i+1}_{zz,k|k}.\\
\end{aligned}
\end{equation*}
Hence, $\hat s^{i+1}_{k|k}$ and $\hat S^{i+1}_{k|k}$ are updated as
\begin{equation*}
\begin{aligned}
& \hat s^{i+1}_{k|k} = \hat s_{k|k-1}+1\\
& \hat S^{i+1}_{k|k} = \hat S_{k|k-1} + B^{i+1}_k.
\end{aligned}
\end{equation*}
For further iteration, define
\begin{equation*}
\begin{aligned}
(\tilde R^{i+1}_{k})^{-1}&=   E^{i+1}\{R^{-1}_{k} \} = \hat s^{i+1}_{k|k}(\hat S^{i+1}_{k|k})^{-1}.
\end{aligned}
\end{equation*}

\vspace{4pt}
\subsubsection{The update of $\lambda_k$ and $\mu_{k}$}Let $\theta = \lambda_{k}$, and it follows from (\ref{eq log pLambda gamma}) that
\begin{equation}
\begin{aligned}
\text{log} ~ q^{i+1}(\lambda_{k})=&\sum_{j=1}^M \lambda_{j,k} \tau^{i+1}_{j,k}+\sum_{j=1}^M  \lambda_{j,k} E^i\{ \text{log}\mu_{j,k}\}+c_{\lambda},
\end{aligned}
\end{equation}
where $c_{\lambda}$ is a constant independent of $\lambda_{k}$, and
\begin{equation*}
\begin{aligned}
\tau^{i+1}_{j,k}=&0.5 \hat g_{j,k|k-1} \text{log}(|\hat G_{j,k|k-1}|)-0.5\text{tr}(\hat G_{j,k|k-1}E^{i+1}\{P^{-1}_{k|k-1}\}\\
&-0.5(\hat g_{j,k|k-1}+n+1)E^{i+1}\{ \text{log}(|P_{k|k-1}|)\}\\
&-0.5(n\hat g_{j,k|k-1})\text{log}2-\text{log}\Gamma_n(\hat g_{j,k|k-1}/2),
\end{aligned}
\end{equation*}
\begin{equation*}
\begin{aligned}
E^{i+1}\{ \text{log}(|P_{k|k-1}|)\}=\text{log}(\hat G^{i+1}_{k|k})-n\text{log}2-\psi_n(0.5\hat g^{i+1}_{k|k}),
\end{aligned}
\end{equation*}
and
\begin{equation*}
\begin{aligned}
&E^{i}\{ \text{log}\mu_{j,k}\}=\psi( \alpha^{i}_{j,k|k})-\psi\bigg(\sum_{j=1}^{M}\alpha^{i}_{j,k|k}\bigg),
\end{aligned}
\end{equation*}
in which $\psi(\cdot)$ is the digamma function.

\vspace{4pt}
Define a new variable
\begin{equation}\label{eq chij}
\begin{aligned}
\chi^{i+1}_{j,k}=& \text{exp}(\tau^{i+1}_{j,k} + E^i\{ \text{log}\mu_{j,k}\}).
\end{aligned}
\end{equation}
Based on the categorical distribution, $\hat \chi^{i+1}_{k}$ is updated by
\begin{equation*}
\begin{aligned}
\hat \chi^{i+1}_{k}&= \chi^{i+1}_{k}/\sum_{j=1}^{M}\chi^{i+1}_{j,k}.
\end{aligned}
\end{equation*}
Similarly, according to the Dirichlet distribution, $ \alpha^{i+1}_{k|k}$ is updated by
\begin{equation*}
\begin{aligned}
& \alpha^{i+1}_{k|k}= \alpha_{k|k-1} + E^{i+1}\{\lambda_k\},
\end{aligned}
\end{equation*}
where $E^{i+1}\{\lambda_k\} = \hat\chi^{i+1}_k$.

\subsection{Variational Approximation: $\gamma_k = 1$}
This subsection studies the case where $z_k$ is transmitted by the sensor at step $k$, i.e., $\gamma_k = 1$. The following unknown parameters will be jointly estimated:
\begin{equation}
\begin{aligned}
\Lambda_1 \triangleq \{x_k,P_{k|k-1},R_k,\lambda_k,\mu_k\}.
\end{aligned}
\end{equation}
The joint pdf $p(\Lambda_1,z_k,\gamma_k)$ is factorized as
\begin{equation*}
\begin{aligned}
p(\Lambda_1,z_k,\gamma_k|\mathcal{I}_{k})=&p(\gamma_k|z_k, x_k, Y_k,\mathcal{I}_{k-1})
 p(x_k|P_{k|k-1},\mathcal{I}_{k-1}) \\ &\times p(z_k|x_k, R_k, \mathcal{I}_{k-1}) p(P_{k|k-1}|\lambda_k)\\
 &\times p(\lambda_k|\mu_k)p(\mu_k)
p(R_k),
\end{aligned}
\end{equation*}
where $ p(\gamma_k=1|z_k, x_k, Y_k,\mathcal{I}_{k-1})= 1- \text{exp}(-\frac{1}{2}e^T_kY_ke_k)$.

\vspace{4pt}
The techniques are similar to that in the above subsection, and the results about  $P_{k|k-1}, \lambda_k, \mu_k$ are the same as that for the case  of $\gamma_k = 0$ situation. Here, the results about $x_k$ and $R_k$ are given.

\vspace{4pt}
\subsubsection{The update of $x_k$} Let $\theta = x_k$, and $\hat x^{i+1}_{k|k}$ is updated by
\begin{equation}\label{eq gamma1 x update}
\begin{aligned}
&\hat x^{i+1}_{k|k}=\hat x_{k|k-1} +K^{i+1}_k(z_k-H_k\hat x_{k|k-1})\\
& K^{i+1}_{k} = \tilde P^i_{k|k-1}H^T_k(H_k\tilde P^{i}_{k|k}H_k+\tilde R^{i}_k)^{-1}\\
& P^{i+1}_{k|k}=\tilde P^{i}_{k|k-1}-\tilde P^{i}_{k|k-1}H^T_k(H_k\tilde P^{i}_{k|k-1}H^T_k+\tilde R^{i}_k)^{-1}H_k\tilde P^{i}_{k|k-1}.\\
\end{aligned}
\end{equation}

\vspace{4pt}
\subsubsection{The update of $R_k$} Let $\theta = R_{k}$, and $\hat s^{i+1}_{k|k}$ and $\hat S^{i+1}_{k|k}$ are updated by
\begin{equation*}
\begin{aligned}
& \hat s^{i+1}_{k|k} = \hat s_{k|k-1}+1\\
& \hat S^{i+1}_{k|k} = \hat S_{k|k-1}+B^{i+1}_k,
\end{aligned}
\end{equation*}
where
\begin{equation*}
\begin{aligned}
B^{i+1}_k=(z_k-H_k\hat x^{i+1}_{k|k})(z_k-H_k\hat x^{i+1}_{k|k})^T+H_kP^{i+1}_{k|k}H^T_k.
\end{aligned}
\end{equation*}

\subsection{Prior Parameters}
This subsection provides a designing method for the prior parameters $\alpha_{k|k-1}$, $\hat s_{k|k-1}$, and $\hat S_{k|k-1}$.
Similarly to \cite{huang2017novel}, the prior parameters $\alpha_{k|k-1}$, $\hat s_{k|k-1}$, and $\hat S_{k|k-1}$ are chosen as
\begin{equation}\label{eq rho prior}
\begin{aligned}
&\alpha_{k|k-1} = \rho \alpha_{k-1|k-1},~~~~\hat s_{k|k-1}=\rho\hat s_{k-1|k-1},\\
&\hat S_{k|k-1} = \rho \hat S_{k-1|k-1},
\end{aligned}
\end{equation}
where $\rho$ is the forgetting factor over $(0,1]$.

\vspace{4pt}
Before the iteration steps, some parameters are initialized as
\begin{equation}\label{eq initial}
\begin{aligned}
& \hat x^0_{k|k}~=~\hat x_{k|k-1},
~~~~~~~~~~~~~~~~~~\hat \chi^0_{k}~~=  \alpha_{k|k-1}/\hat \alpha_{k|k-1}, \\
&\hat g^{0}_{k|k-1}= \sum_{j=1}^M\hat \chi^0_{j,k}\hat g_{j,k|k-1},
~~~\hat G^{0}_{k|k-1}= \sum_{j=1}^M\hat \chi^0_{j,k}\hat G_{j,k|k-1},\\
&\tilde P^0_{k|k-1}= \hat G^{0}_{k|k-1}/\hat g^{0}_{k|k-1}, ~~~~~~\tilde R^0_{k}~~~ = ~~~\hat S^0_{k|k-1}/\hat s^0_{k|k-1}.
\end{aligned}
\end{equation}

\vspace{4pt}
Now, based on the above results, the event-triggered variational Bayesian filter (\textbf{ETVBF}) is summarized as Algorithm \ref{algorithm event-triger}.
It is observed that only $x_k$ and $R_k$ are directly affected, when the measurement information is not received by the estimator. Then, other parameters are influenced  by $x_k$ and $R_k$.

\begin{algorithm}[t]
\caption{Event-triggered Variational Bayesian Filter (ETVBF)}
\label{algorithm event-triger}
\hspace*{0.01in}{\bf Input:} $\hat x_{k-1|k-1}$, $P_{k-1|k-1}$, $\hat s_{k-1|k-1}$, $\hat S_{k-1|k-1}$, $\alpha_{k-1|k-1}$, $\hat g_{j,k|k-1}$, $\bar Q_{j,k}$, $\rho$, $\delta$, $N$.

\hspace*{0.01in}\textbf{Predicted state:}
\begin{algorithmic}
\STATE $\hat x_{k|k-1}~~=~~F_{k-1}\hat x_{k-1|k-1}$,
 \STATE $P_{j,k|k-1}
=~~F_{k-1}P_{k-1|k-1}F^T_{k-1}+\bar Q_{j,k}$,
\STATE $\hat G_{j, k|k-1} = ~~\hat g_{j,k|k-1} P_{j,k|k-1}$.
\end{algorithmic}

\hspace*{0.01in}\textbf{Updated state:}
\begin{algorithmic}
\STATE\hspace*{-0.04in}\textbf{Initialization:} \STATE Initialized as (\ref{eq rho prior}) and (\ref{eq initial}).
\STATE\hspace*{-0.04in}\textbf{For $i=0:N-1$}
\STATE \hspace*{-0.03in}\textbf{Update $x,z$:}
\STATE \hspace*{0.03in}\textbf{If} $\gamma_k = 0$,
\STATE \hspace*{0.03in}$\hat x^{i+1}_{k|k} = \hat x_{k|k-1}$,~~~$\hat z^{i+1}_{k|k} = H_k\hat x_{k|k-1}$,
\STATE \hspace*{0.03in}$P^{i+1}_{k|k}$,  $P^{i+1}_{zz,k|k}$, and $P^{i+1}_{xz,k|k}$ are calculated as (\ref{eq Pkk}), (\ref{eq Pzzkk}), and(\ref{eq Pxzkk}), respectively.

\STATE \hspace*{0.03in}\textbf{Else if} $\gamma_k = 1$,
\STATE \hspace*{0.03in}$K^{i+1}_{k}$, $\hat x^{i+1}_{k|k}$, and $P^{i+1}_{k|k}$ are updated  as (\ref{eq gamma1 x update}).
\STATE \hspace*{0.03in}\textbf{End}

\STATE \hspace*{-0.03in}\textbf{Update} $\text{IW}(P_{k|k-1}|\hat g^{i+1}_{k|k},\hat G^{i+1}_{k|k})$:

\STATE \hspace*{0.03in}$\hat g^{i+1}_{k|k}= \sum_{j=1}^M\hat\chi^i_{j,k}\hat g_{j,k|k-1}+1$,
\STATE \hspace*{0.03in}$\hat G^{i+1}_{k|k}= \sum_{j=1}^M\hat\chi^i_{j,k}\hat G_{j,k|k-1}+A^{i+1}_k$,
\STATE \hspace*{0.03in}$A^{i+1}_k=P^{i+1}_{k|k}+(\hat x^{i+1}_{k|k}-\hat x_{k|k-1})(\hat x^{i+1}_{k|k}-\hat x_{k|k-1})^T$.

\STATE \hspace*{-0.03in}\textbf{Update} $\text{IW}(R_k|\hat s_{k|k-1},\hat S_{k|k-1})$:
\STATE \hspace*{0.03in} $\hat s^{i+1}_{k|k-1}= \hat s_{k|k-1} + 1$, ~~~$\hat S^{i+1}_{k|k-1} = \hat S_{k|k-1}+B^{i+1}_{k}$,
\STATE \hspace*{0.03in}\textbf{If} $\gamma_k = 0$, \STATE \hspace*{0.03in}$B^{i+1}_k= H_kP^{i+1}_{k|k}H^T_k - (H_kP^{i+1}_{xz,k|k})^T - H_kP^{i+1}_{xz,k|k}+P^{i+1}_{zz,k|k}$,
\STATE \hspace*{0.03in}\textbf{Else if} $\gamma_k = 1$,
\STATE \hspace*{0.03in}$B^{i+1}_k=(z_k-H_k\hat x^{i+1}_{k|k})(z_k-H_k\hat x^{i+1}_{k|k})^T+H_kP^{i+1}_{k|k}H^T_k$.
\STATE \hspace*{0.03in}\textbf{End}

\STATE \hspace*{-0.03in}\textbf{Update $\lambda_k$:}
\STATE \hspace*{0.03in}$E^{i+1}\{ \text{log}\mu_{j,k}\}=\psi( \alpha^{i}_{j,k|k})-\psi(\sum_{j=1}^{M}\alpha^{i}_{j,k|k})$,
\STATE \hspace*{0.03in}$E^{i+1}\{ \text{log}(|P_{k|k-1}|)\}=\text{log}(\hat G^{i+1}_{k|k})-n\text{log}2-\psi(0.5\hat g^{i+1}_{k|k}) $,
\STATE \hspace*{0.03in}$\chi^{i+1}_{j,k}$ is computed as (\ref{eq chij}),
\STATE \hspace*{0.03in}
$\hat\chi^{i+1}_{k}= \chi^{i+1}_{k}/\sum_{j=1}^{M}\chi^{i+1}_{j,k}$.

\STATE \hspace*{-0.03in}\textbf{Update $\mu_k$:}
\STATE \hspace*{0.03in}$\alpha^{i+1}_{k|k}= \alpha_{k|k-1} + \hat \chi^{i+1}_{k}$.

\STATE\hspace*{-0.03in}\textbf{If} $||\hat x^{i+1}_{k|k}-\hat x^i_{k|k}||/ ||\hat x^i_{k|k}||\leq \delta,$ terminate iteration.

\STATE  \hspace*{-0.04in}\textbf{end}
\end{algorithmic}

\hspace*{0.01in}\textbf{Output:}
\hspace*{0.03in}$\hat x_{k|k} = \hat x^{i+1}_{k|k}$, $P_{k|k}=P^{i+1}_{k|k}$, $\hat s_{k|k}=\hat s^{i+1}_{k|k}$, $\hat S_{k|k}=\hat S^{i+1}_{k|k}$, $\alpha_{k|k}=\alpha^{i+1}_{k|k}$.
\end{algorithm}

\subsection{Discussions}
In the algorithm, there exist several parameters that need to be selected and designed, i.e., $\hat g_{j,k|k-1}$, $\bar Q_{j,k}$, $\alpha_{0|0}$,  $\rho$, $R_0$, $\hat s_{0|0}$, $N$, and $\delta$. In the following, the influence of these parameters, the selections of these parameters, and the performances are discussed.

\vspace{4pt}
\subsubsection{The influence of parameters}
~

For parameters $\hat g_{j,k|k-1}$ and $\bar Q_{j.k}$,
based on Algorithm \ref{algorithm event-triger}, $\tilde P^{i+1}_{k|k-1}$ can be formulated as
\begin{equation}
\begin{aligned}
\tilde P^{i+1}_{k|k-1}
&= \frac{\sum_{j=1}^M\hat\chi^i_{j,k}\hat g_{j,k|k-1}P_{j,k|k-1}+A^{i+1}_k}{\sum_{j=1}^M\hat\chi^i_{j,k}\hat g_{j,k|k-1}+1},
\end{aligned}
\end{equation}
where  $P_{j,k|k-1}=F_{k-1}P_{k-1|k-1}F^T_{k-1}+\bar Q_{j,k}$.

\vspace{4pt}
 a) $\hat g_{j,k|k-1}$: It can be observed that $\tilde P^{i+1}_{k|k-1}$ is a weighted sum of $P_{j,k|k-1}, j\in M$ and $A^{i+1}_k$
by  $\hat\chi^i_{j,k}\hat g_{j,k|k-1}$ and $1$. The bigger the $\hat g_{j,k|k-1}$ is, the more the prior information $P_{j,k|k-1}$ is introduced into $\tilde P^{i+1}_{k|k-1}$.

b) $\bar Q_{j,k}$: When $\hat g_{j,k|k-1}, j\in M$, are set as the same value $\tilde g_{k|k-1}$,
 $\sum_{j=1}^M\hat\chi^i_{j,k}\hat g_{j,k|k-1}P_{j,k|k-1}=\tilde g_{k|k-1}F_{k-1}P_{k-1|k-1}F^T_{k-1}+\tilde g_{k|k-1}\sum_{j=1}^M\hat\chi^i_{j,k}\bar Q_{j,k}$.
The adaptive parameter
$\hat\chi^i_{j,k}$ is subject to $\sum_{j=1}^M\hat\chi^i_{j,k}=1$, $0\leq \hat\chi^i_{j,k} \leq 1$, which means that a  convex combination of
multiple nominal process noise covariances $\bar Q_{j,k}, j\in M$,  is adaptively  achieved. Multiple nominal process noise covariances are characterized by a range of the true covariance, and
 their accuracy requirements are reduced compared to the  single covariance case.
 This condition is easier to  satisfy and verify.

\vspace{4pt}
For parameters $\rho$, $R_0$, and $\hat s_{0|0}$, similar to \cite{huang2017novel},  it can be obtained that
\begin{equation}\label{eq tilde R analysis 1}
\begin{aligned}
\tilde R^{i+1}_{k}&= \frac{\eta(\rho,k)\hat R_{k-1|k-1}+B^{i+1}_{k}}{\eta(\rho,k) + 1},
\end{aligned}
\end{equation}
\begin{equation}\label{eq tilde R analysis 2}
\begin{aligned}
\hat R_{k-1} &=  \bigg(\prod^{k-1}_{i=1}q_i\bigg)\hat R_0 + \sum^{k-1}_{i=1}\bigg(\prod^{k-1}_{j=i+1}q_j\bigg)\tilde B_i,
\end{aligned}
\end{equation}
where $\eta(\rho, k) = \rho^{k}\hat s_{0|0} + \frac{\rho^{k}-\rho}{\rho-1}$, $\rho\in(0~1]$,
 $q_k = \frac{\eta(\rho,k)}{\eta(\rho,k) + 1}$,  $\tilde B_{k} = \frac{B^{stop}_{k}}{\eta(\rho,k) + 1}$, and $B^{stop}_{k}$ represents  $B^{i}_{k}$ at the loop termination step.

c) $\rho$:
For \textbf{ETVBF},  $B^{i+1}_{k}$ in (\ref{eq tilde R analysis 1})  is different  under different event-triggered  situations, and   $\eta(\rho, k)$  plays a role in balancing a weighted  sum of $\hat R_{k-1|k-1}$ and $B^{i+1}_{k}$. Moreover, $\eta(\rho, k)$ is a monotone increasing function of $\rho$ as $k$ approaches infinity, and $\lim \limits_{k \to \infty} \eta(\rho, k) =  \frac{-\rho}{\rho-1}$ \cite{huang2017novel}. Hence, the smaller the $\rho$ is, the more information the $B^{i+1}_k$ is introduced into $\tilde R^{i+1}_{k}$.

d) $\hat R_0$:
Based on (\ref{eq tilde R analysis 2}), it can be concluded that
the effect of the past information about $\hat R_0$ and $\tilde B_k$  exponentially decays   as the time increases.

e) $N$ and  $\delta$:
Based on the variational inference theory (\ref{eq max ELBO}) and the fixed-point iteration (\ref{eq compute KL iteration})\cite{hildebrand1987introduction}, a more accurate  approximate optimal solution can be  obtained by increasing the iteration number $N$. By setting a reasonable loop termination condition $\delta$,   the iteration can be terminated early with guaranteed  accuracy.

\vspace{4pt}
\subsubsection{The selections of parameters}
For parameters $\hat g_{j,k|k-1}$  and $\hat s_{0|0}$, they can be selected within a wide range, and it is suggested that they are tuned based on the specific system.
For $\bar Q_{j,k}$,
to obtain better performances,  it is suggested that multiple nominal process noise covariances $\bar Q_{j,k}$ are selected such that
the true process noise covariance $Q_k$ satisfies $\text{min}( \bar Q_{j,k}, j\in M) \leq Q_k\leq \text{max}( \bar Q_{j,k}, j\in M)$. For $\alpha_{0|0}$, it can be set as $1_{1\times M}$, since it will be  adaptively adjusted. For $\rho$,
the forgetting factor  is suggested to be selected as $\rho \in [0.94,~1]$ to obtain a good performance, as shown in the simulation part. For $N$ and $\delta$,  they are selected as a big number and
a sufficiently small number, respectively, to guarantee the iteration performance. It is worth mentioning that the effectiveness of such selections is verified in the simulation part.

\vspace{4pt}
\subsubsection{Performance  Analysis}
 The estimation  performance  has a strong correlation with the deviation between  the nominal  value and the true value, such as $\bar R_0$ and $R_0$.
When the deviation is smaller, the proposed algorithm has a better performance, as shown in Fig. \ref{fig yr10to300rmse} in Simulations.
 The parameters $N$ and $\delta$ have a main  influence on  the estimation performance, and they have been discussed above.
A concept ``numerical stability" from  \cite{huang2017novel}\cite{hildebrand1987introduction} is introduced to guarantee the feasibility of the proposed  algorithm, i.e., the covariance matrices  must be  positive definite  when  the algorithm is in the numerical iteration process.
Since $\bar Q_{j,k}>0$, $\hat R_0>0$, $A^{i+1}_k>0$, and $B^{i+1}_k>0$,  it follows that $\tilde P^{i+1}_{k|k-1}>0$ and $\tilde R^{i+1}_{k}>0$.  Thus,  \textbf{ETVBF} is numerically stable.

\vspace{4pt}
\begin{proposition}
For systems (\ref{eq dynamical equation})-(\ref{eq measurement equation}) under Algorithm \ref{algorithm event-triger},
the estimation error $e_{k|k}= \hat x_{k|k}-x_k$ is Gaussian.

\end{proposition}
\begin{proof}
In Algorithm \ref{algorithm event-triger}, denote the measurement update error and  the predicted error as $e_{k|k}= \hat x_{k|k}-x_k$ and $e_{k|k-1}=\hat x_{k|k-1}-x_k$, respectively.
When $\gamma_k=0$, one has $e_{k|k}= \hat x_{k|k-1}-x_k=e_{k|k-1}$.
When $\gamma_k =1$, one has $e_{k|k}=(I-K^{i+1}_kH_k)(\hat x_{k|k-1}-x_k) + K^{i+1}_{k}\nu_k=(I-K^{i+1}_kH_k)e_{k|k-1}+ K^{i+1}_{k}\nu_k$. Additionally, $e_{k|k-1}=F_k\hat x_{k-1|k-1}-F_kx_{k-1}-\omega_{k-1}=F_ke_{k-1|k-1}-\omega_{k-1}$.
 Due to the Gaussianity of $e_{0|0}=\hat x_{0|0}-x_0$, $\omega_{k-1}$, and $\nu_k$, the error $e_{k|k}$  is Gaussian as can be verified inductively.
\end{proof}

\section{Simulations}
 This section shows some simulations and  comparison with existing algorithms to verify the effectiveness  of the proposed {\bf ETVBF}.

\vspace{4pt}
A tracking problem for a vehicle is considered as in \cite{huang2017novel}. The dynamical system is described by (\ref{eq dynamical equation}) and (\ref{eq measurement equation}) with system matrices
 \begin{equation}
\begin{aligned}
&F_k = \left[\begin{array}{cc}
I_2  & TI_2  \\
0  & I_2  \\
\end{array} \right]
\end{aligned}
\end{equation}
 and
  \begin{equation}
\begin{aligned}
&H_k = \left[\begin{array}{cccc}
1  & 0 & 0 & 0 \\
0  & 1 & 0 & 0 \\
\end{array} \right],
\end{aligned}
\end{equation}
 respectively,  where $T=1$s. The state dimension  and the measurement dimension are $n=4$ and $m=2$, respectively.
The true  process  and the true measurement noise covariances are set as
\begin{equation*}
\begin{aligned}
&Q_k = (6+0.5\text{cos}((\pi k)/{T_f}))\left[\begin{array}{cccc}
T^3/3 & 0 & T^2/2 & 0 \\
0 & T^3/3 & 0 & T^2/2 \\
T^2/2 & 0 & T & 0 \\
0 & T^2/2 & 0 & T\\
\end{array} \right]
\end{aligned}
\end{equation*}
and
\begin{equation*}
\begin{aligned}
&R_k = (100+50\text{cos}((\pi k)/T_f))\left[\begin{array}{cc}
1  & 0.5  \\
0.5  & 1  \\
\end{array} \right],
\end{aligned}
\end{equation*}
respectively, where $T_f =500$.

\vspace{4pt}
In this section, the performances of \textbf{ETVBF} in Algorithm \ref{algorithm event-triger},  a closed-loop stochastic event-triggered Kalman filter (\textbf{CLSET-KF}) in \cite{han2015stochastic}, and a variational Bayesian filter (\textbf{VBF}, without the triggering mechanism in Algorithm 1),  a variational Bayesian based adaptive Kalman filter (\textbf{VBAKF}) in \cite{huang2017novel}, and a variational Bayesian adaptive Kalman filter with Gaussian-Gamma mixture (\textbf{VBAKF-GGM}) in \cite{zhu2021adaptive} are compared.

\vspace{4pt}
The nominal initial measurement  noise covariance is set as  $\bar R_0 = rI_2$, and the event-triggered parameter is set as $Y_k=yI_2$, where $r$ and $y$ are the noise scale factor and the event-triggered scale factor, respectively. The step internal $N_{step}$ of $k$ and  the total iteration number $N$ are chosen  as $150$ and $50$, respectively.
The initial state is $x_0= [100, 100, 10, 10]^T$, and the initial estimation error covariance is $\hat P_{0|0}= 100I_4$. Then, the initial estimate state is given by $\hat x_{0|0} \sim N(x_0, \hat P_{0|0})$.
The dof parameters  are  set as $\hat g_{k|k} = 10\times 1_{1\times M}$ and $\hat s_{0|0} = 5$.
 The forgetting factor and the initial concentration parameter  are selected as
$\rho = 0.997$ and
$\hat \alpha_0 = 1_{1\times 5}$, respectively. For every case,
Monte Carlo simulation experiments with $N_{MC}=500$ are performed.
The mixture term is selected as $M = 5$, and  the nominal process  noise covariances are set as $\bar Q_{k,1} = I_4, \bar Q_{k,2} = 2I_4, \bar Q_{k,3} = 3I_4, \bar Q_{k,4} = 9I_4, \bar Q_{k,5} = 10I_4$.
 For \textbf{VBAKF-GGW}, the shape parameters are set as $a_0=e_0=[10, 10, 10, 10]$, and  the rate parameters are set as $b_0=f_0=[10, 100, 1000, 10000]$.
For \textbf{CLSET-KF}, \textbf{VBAKF}, and \textbf{VBAKF-GGW}, the nominal process noise covariance is set as  $\bar Q_k = 4I_4$.

\vspace{4pt}
The root-mean-square error (RMSE) is utilized to evaluate the performance of the algorithm:
\begin{equation*}
\text{RMSE}=\sqrt{\frac{1}{nN_{MC}N_{step}}\sum^{N_{MC}}_{k=1}\sum^{N_{step}}_{j=1}\sum^{n}_{l=1}(\hat x_{k|k,j}(l)-x_{k,j}(l))^2 },
\end{equation*}
where $n$, $N_{MC}$, $N_{step}$ denote  the dimension of the state, the total Monte Carlo experiment number, and the step internal, respectively, and
$\hat x_{k|k,j}(l)$  and  $x_{k,j}(l)$  represent the $l$-th component  of the vector $\hat x_{k|k}$ and the vector $x_k$  at the $k$ step in the $j$-th trail experiment, respectively.
Similarly, the communication rate $\gamma_{trail}$ is defined as
\begin{equation*}
\gamma_{trail}=\sqrt{\frac{1}{N_{MC}N_{step}}\sum^{N_{MC}}_{k=1}\sum^{N_{step}}_{j=1}\gamma_{k,j}},
\end{equation*}
where $\gamma_{k,j}$ denotes $\gamma_k$ of the $j$-th trail experiment at step $k$.

\vspace{4pt}
First, under the  different event-triggered  communication rates,  the estimation performances are compared among  five algorithms.  The event-triggered communication rates are dominated by $Y_k$. Hence, the event-triggered scale factor $y$ is set as $0.0005:0.0005:0.1$, and the noise covariance scale factor  is set as $r=150$.  Fig. \ref{fig yr150rmse} and  Fig. \ref{fig yr150gamma} show the RMSE of the five algorithms
and the event-triggered communication rates of \textbf{ETVBF} and \textbf{CLSET-KF}, respectively. To show the event-triggered performance, denote $t_{l}$ as the $l$-th event-triggered time instant, and $t_{l+1}-t_{l}$ as the interval between two adjacent event-triggered time instants.
Fig. \ref{fig triggeredtime}  presents the scheduling sequences of \textbf{ETVBF} and \textbf{CLSET-KF} under the event-triggered mechanism with $y=0.0005$ and $r=150$.
Fig. 5 shows the average iteration numbers of \textbf{ETVBF}, \textbf{VBF}, and \textbf{VBAKF-GGM} with $y=0.0005$ and $r=150$.

\begin{figure}[!htb]
\centering
{\includegraphics[width=3.2in]{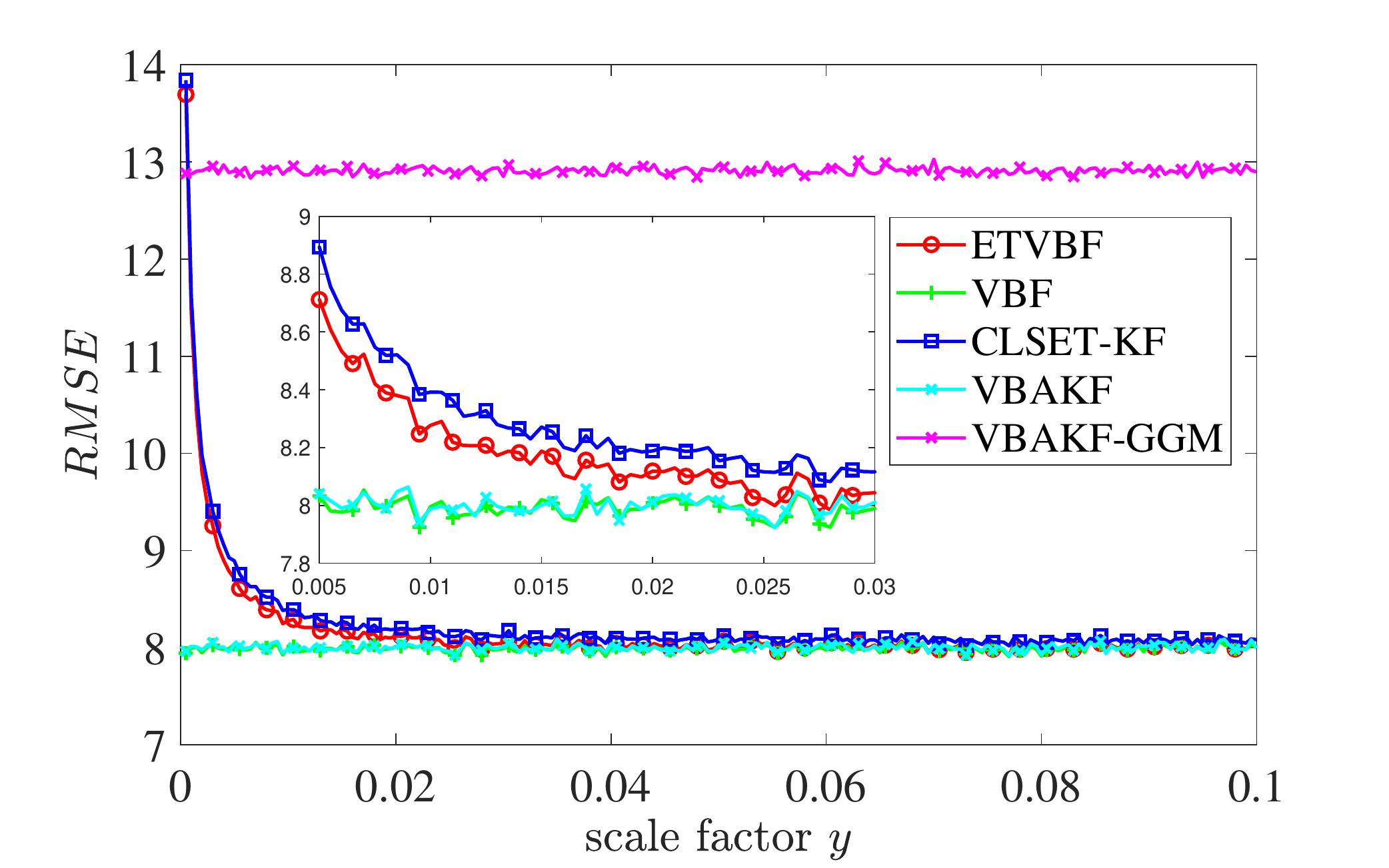}}
\caption{The RMSE of ETVBF, VBF, CLSET-KF, VBAKF, and VBAKF-GGM.}
\label{fig yr150rmse}
\end{figure}

\begin{figure}[!htb]
\centering
{\includegraphics[width=3.2in]{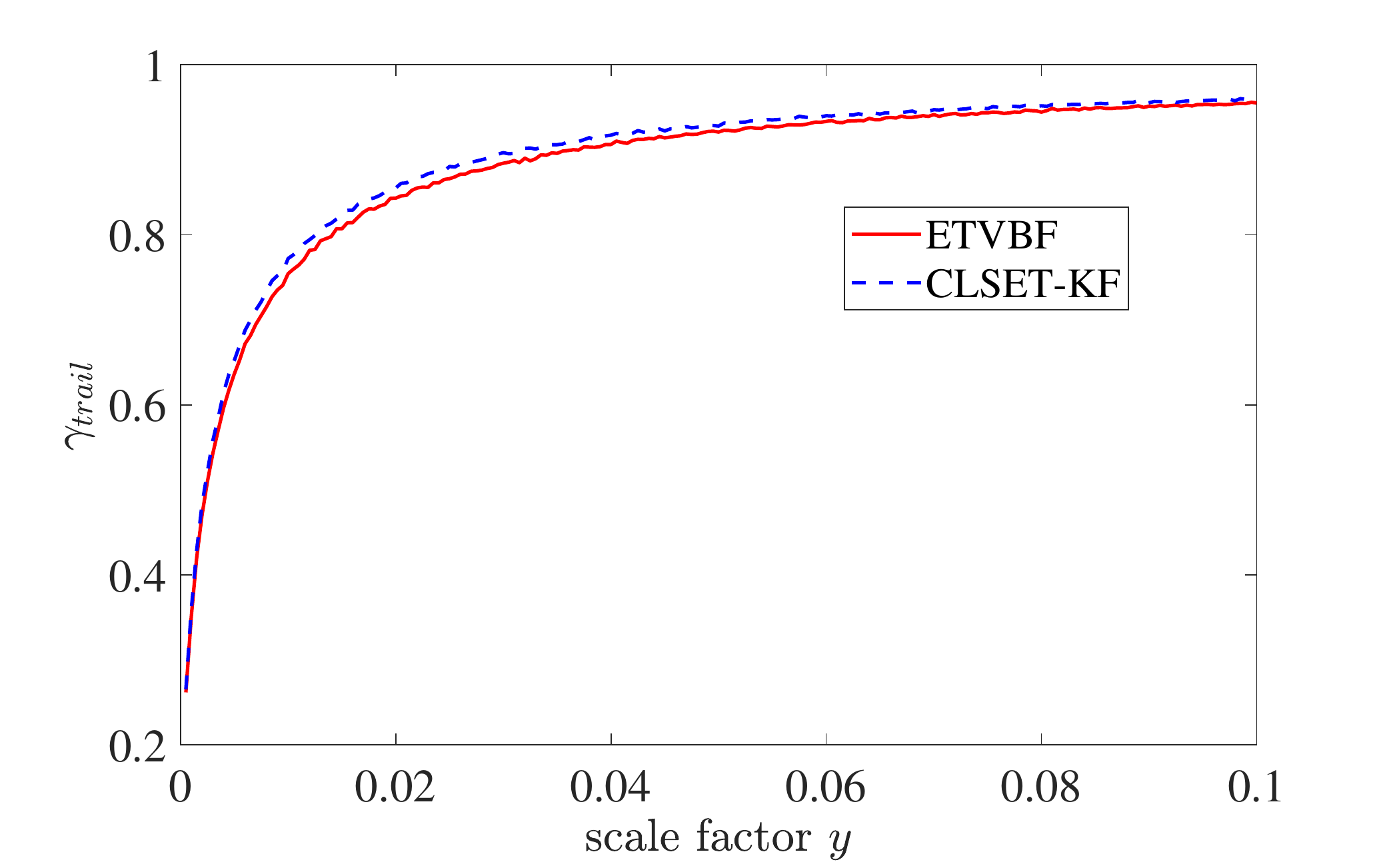}}
\caption{The communication rates of ETVBF and CLSET-KF.}
\label{fig yr150gamma}
\end{figure}

\vspace{4pt}
\begin{enumerate}
  \item
  \textbf{ETVBF} has a better performance than \textbf{CLSET-KF}  all the time, while \textbf{CLSET-KF} needs higher communication rates than \textbf{ETVBF}.
  \textbf{VBAKF-GGW} has the worst performance, since  it may be specially designed in the presence of outliers.

  \item
  As  $y$ increases, the communication rates of \textbf{ETVBF }and \textbf{CLSET-KF}  increase, and their performances  become better. \textbf{VBF} and  \textbf{VBAKF} have the best performance at a high transmission cost. As $y$ increases,  the performance of \textbf{ETVBF} approaches that of \textbf{VBF}.

  \item
  The proposed  \textbf{ETVBF} needs the least average iteration number, and \textbf{VBF} needs a little larger average iteration number than \textbf{ETVBF}. As these filters converge, the average iteration numbers decrease gradually.

\end{enumerate}

\begin{figure}[!htb]
\centering
{\includegraphics[width=3.2in]{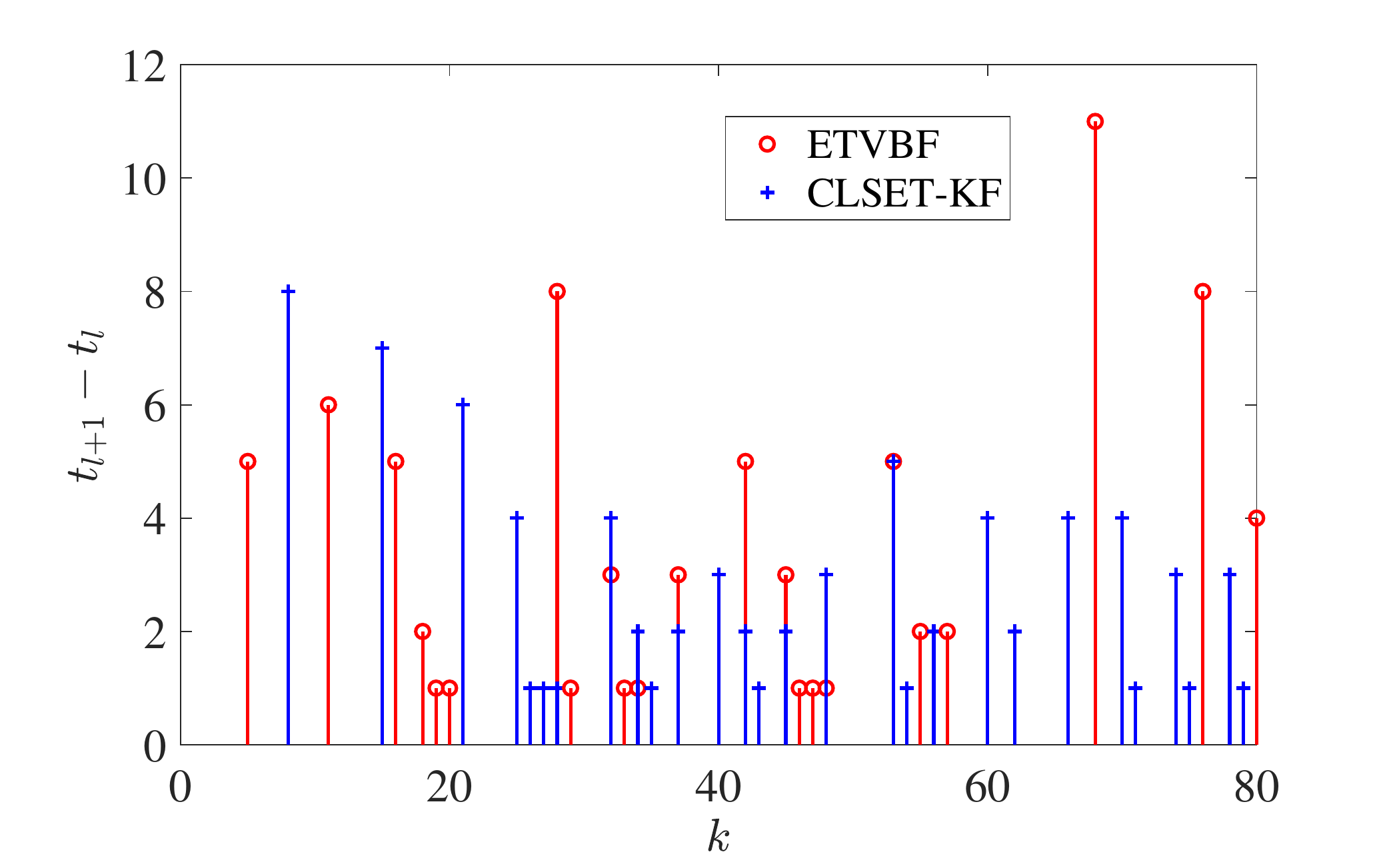}}
\caption{Scheduling sequences  of ETVBF and CLSET-KF under the event-triggered mechanism.}
\label{fig triggeredtime}
\end{figure}

\begin{figure}[!htb]
\centering
{\includegraphics[width=3.2in]{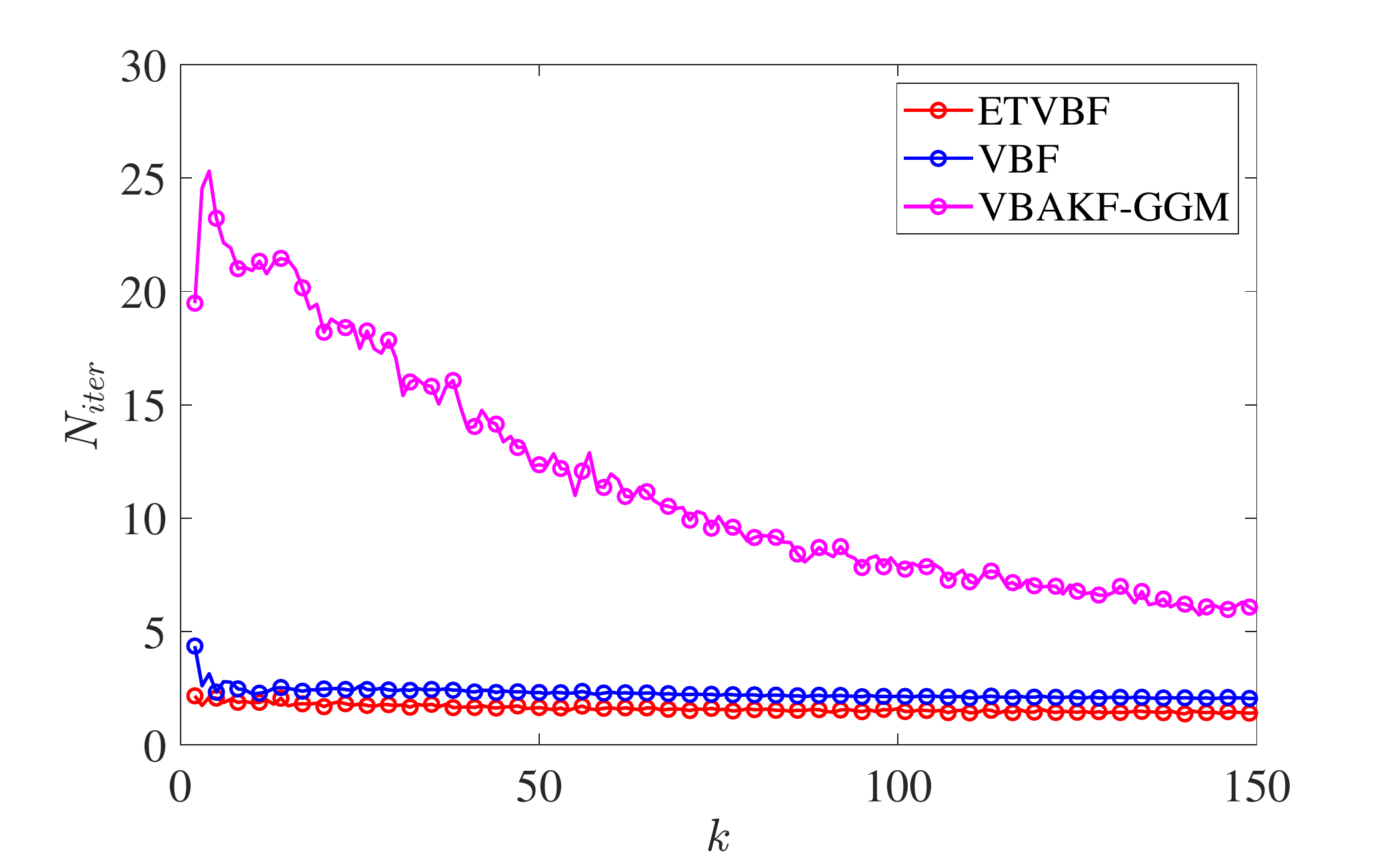}}
\caption{The average iteration numbers of ETVBF, VBF and VBAKF-GGM.}
\label{fig Niter}
\end{figure}

\vspace{4pt}
Next, with different nominal measurement noise covariances $\bar R_0$, the estimation performances are compared among the  five algorithms. Similarly, the noise scale $r$ is set as $10:10:300$, and the event-triggered scale parameter is set as $y=0.015$.
Fig. \ref{fig yr10to300rmse} and Fig. \ref{fig yr10to300gamma} illustrate the RMSE of the five algorithms and communication rates of \textbf{ETVBF} and \textbf{CLSET-KF}, respectively.
\begin{enumerate}
  \item  When the error between the nominal measurement noise covariance $rI_2$ and the true noise covariance is small, all algorithms have good performances. As the error gets larger, the performances are degraded.  \textbf{ETVBF} and \textbf{VBF} have better and stabler performances  for different noise covariances.
  \item
  As the nominal measurement noise covariance  increases, the communication rates of \textbf{ETVBF} and \textbf{CLSET-KF} increase.
\end{enumerate}

\begin{figure}[!htb]
\centering
{\includegraphics[width=3.2in]{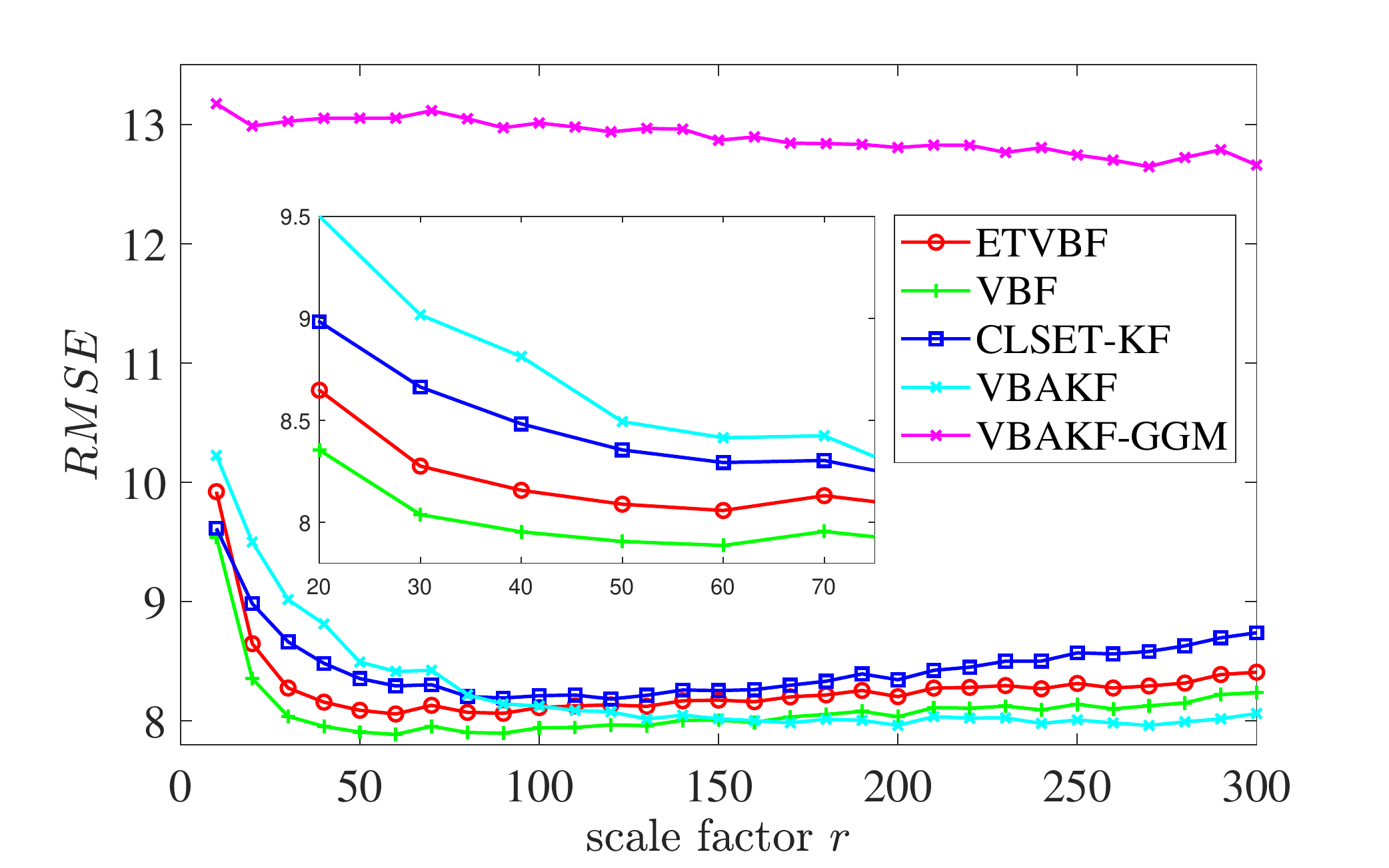}}
\caption{The RMSE of ETVBF, VBF, CLSET-KF, VBAKF, and VBAKF-GGM.}
\label{fig yr10to300rmse}
\end{figure}

\begin{figure}[!htb]
\centering
{\includegraphics[width=3.2in]{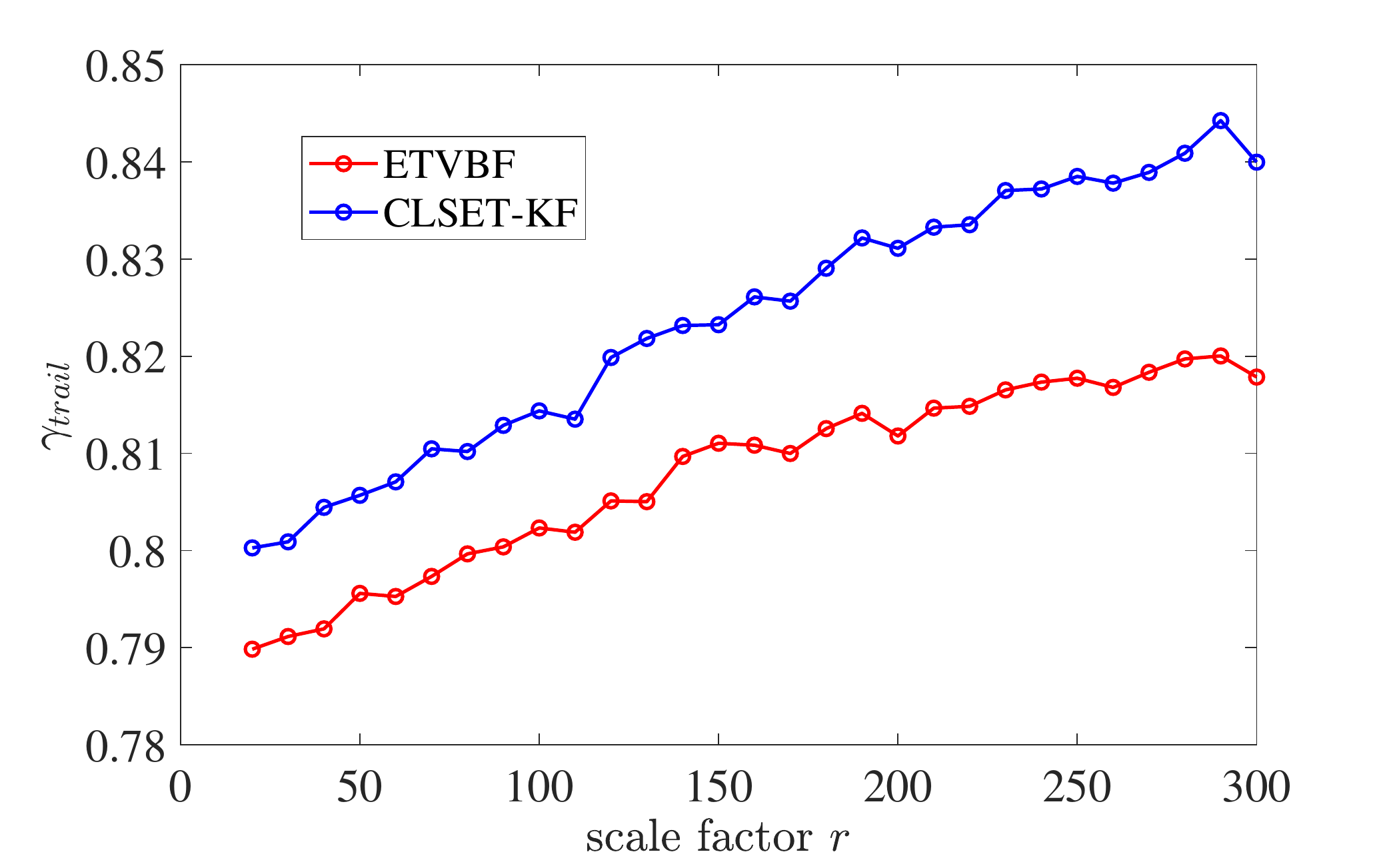}}
\caption{The communication rates of ETVBF and CLSET-KF.}
\label{fig yr10to300gamma}
\end{figure}

\begin{figure}[!htb]
\centering
{\includegraphics[width=3.2in]{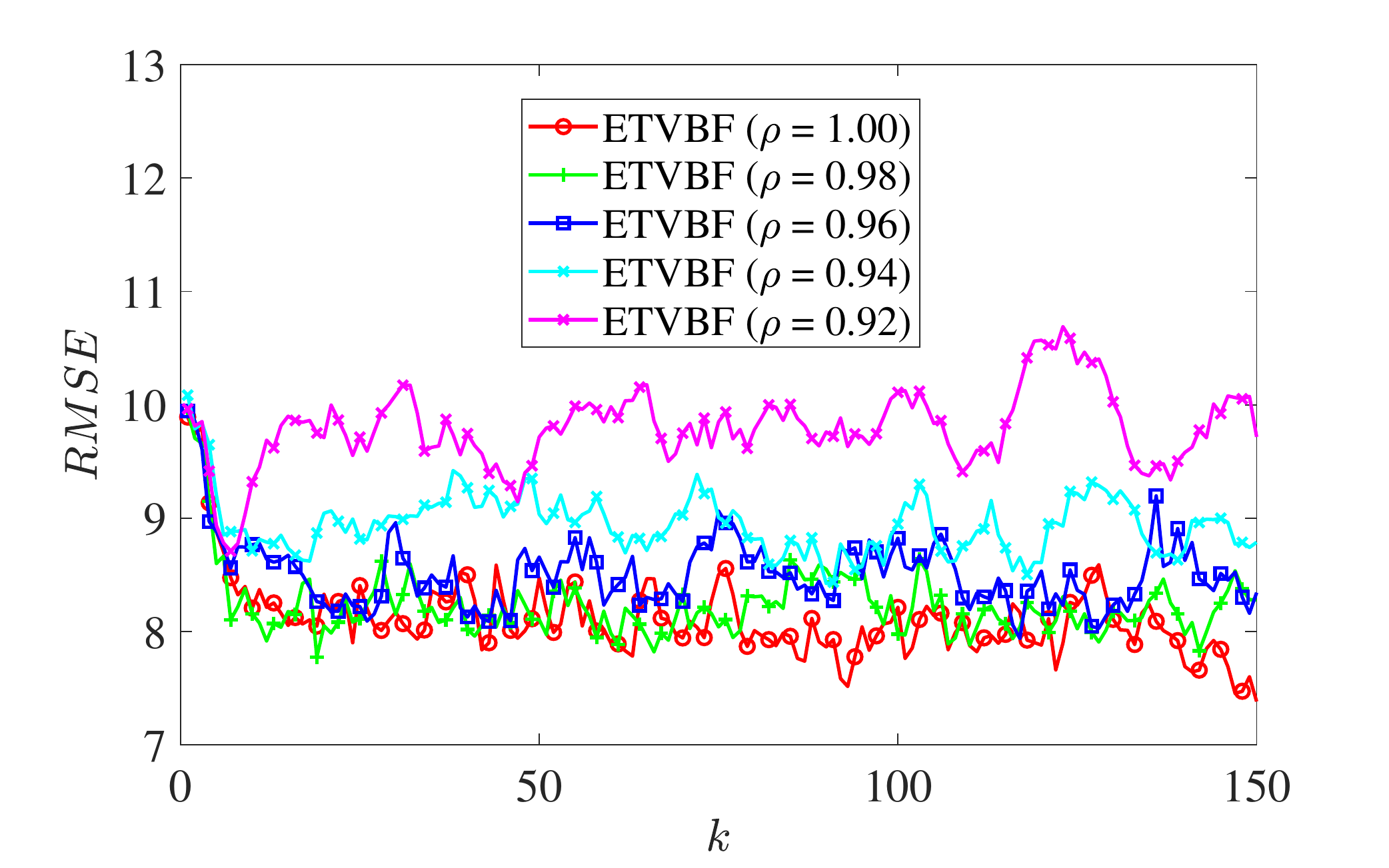}}
\caption{The RMSE of ETVBF when $\rho = 0.92, 0.94, 0.96, 0.98, 1.00$.}
\label{fig rho0.9to1}
\end{figure}

\vspace{4pt}
Finally, under  different  forgetting  factors $\rho$, the estimation performances of \textbf{ETVBF} are compared. Fig. \ref{fig rho0.9to1} shows the RMSE of \textbf{ETVBF} under the forgetting factors $\rho = 0.92, 0.94, 0.96, 0.98, 1.00$. When $\rho<0.94$, the RMSE of \textbf{ETVBF} shows a significant decline. Hence,  it is suggested that  the forgetting  factor is selected as $\rho \in [0.94, ~ 1]$.

\vspace{4pt}
In summary, \textbf{ETVBF} possesses excellent and robust performances.

\section{Conclusion}
In this paper, an event-triggered variational Bayesian filter is proposed for systems with unknown and time-varying noise covariances. The state vector, the predicted error covariance, and the unknown measurement noise covariance are jointly estimated. Simulations show  excellent and robust performances of the proposed algorithm. The event-triggered  variational Bayesian filter in the nonlinear setting  will be considered in the future.

\bibliographystyle{IEEEtran}
\bibliography{ref}

\end{document}